\documentclass[acmsmall,screen]{acmart}
\setcopyright{rightsretained}
\acmPrice{}
\acmDOI{10.1145/3473580}
\acmYear{2021}
\copyrightyear{2021}
\acmSubmissionID{icfp21main-p78-p}
\acmJournal{PACMPL}
\acmVolume{5}
\acmNumber{ICFP}
\acmArticle{75}
\acmMonth{8}

\usepackage{amsfonts}
\usepackage{amsmath}
\usepackage{amsthm}
\usepackage{extarrows}
\usepackage{graphicx}
\usepackage{ifoddpage}
\usepackage{listings}
\usepackage{marginnote}
\usepackage{mathtools}
\usepackage{mathpartir}
\usepackage{semantic}
\usepackage{skull}
\usepackage{todonotes}
\usepackage{xcolor}
\usepackage{xspace}
\usepackage{stmaryrd}

\usepackage{booktabs}   
\usepackage{subcaption} 

\newcounter{WarnCounts}

\newcommand{\mute}[1]{\ifdraftx{#1} \else{} \fi}
\newcommand{\decorateWC}{
  \stepcounter{WarnCounts}
  \checkoddpage
  \ifoddpageoroneside
     \mbox{}\marginnote{\textcolor{red}{$\skull_{\theWarnCounts}$}}
     \else
     \reversemarginpar
     \mbox{}\marginnote{\textcolor{red}{$\skull_{\theWarnCounts}$}}
     \normalmarginpar
   \fi
}

\newcommand{\signedComment}[3]
           {\mute{\decorateWC\textcolor{#2}{(#1: {#3})}}}

\newcommand{\an}[1]{\signedComment{Aleks}{red}{#1}}

\newcommand{\eqdef}{\mathrel{\:\widehat{=}\:}}
\newcommand{\cmtte}{ECMTT\xspace}
\newcommand{\reduces}[1]{\mapsto_{#1}}
\newcommand{\redrel}[2]{#1 \reduces{} #2}

\newcommand{\nhandlertype}[4]{#1 [#2] \xRightarrow{#3} #4}
\newcommand{\kbinding}{\sim\hspace{-0.45em}\vcentcolon\,}
\newcommand{\rotatedsim}{\rotatebox[origin=c]{45}{$\sim$}}

\newcommand{\ectxsym}{\cdot}

\newcommand{\ejdgmt}[1]{\vdash #1}
\newcommand{\enjdgmt}[1]{\nvdash #1}
\newcommand{\jdgmt}[3]{#1; #2 \vdash #3}
\newcommand{\jdgmtsim}[1]{\jdgmt{\Delta}{\Gamma}{#1}}

\newcommand{\pjdgmt}[3]{#1 \vdash \tbind{#2}{#3}}
\newcommand{\pjdgmtsim}[2]{\pjdgmt{\Delta}{#1}{#2}}

\newcommand{\tbind}[2]{#1\,{:}\,#2}
\newcommand{\chypbind}[3]{#1\,{\div}\,#2 \Rightarrow #3}
\newcommand{\cbind}[2]{#1 \div #2}
\newcommand{\khypbind}[4]{#1\,{\kbinding}\,#2 \xRightarrow{#3} #4}
\newcommand{\mhypbind}[3]{#1\,{::}\,#2 [#3]}

\newcommand{\hbind}[5]{#1 \div \nhandlertype{#2}{#3}{#4}{#5}}
\newcommand{\hseqbind}[4]{#1 \div #2 [#3] \Rrightarrow #4}
\newcommand{\sbind}[2]{#1 \div_{\mathsf{s}} #2}
\newcommand{\esbind}[2]{{#2}\,{\shortrightarrow}\,{#1}}

\newcommand{\boolt}[0]{\textit{bool}}
\newcommand{\funt}[2]{#1 \rightarrow #2}

\newcommand{\boxt}[2]{[#1] #2}
\newcommand{\unitt}[0]{\textit{unit}}
\newcommand{\intt}[0]{\textit{int}}

\newcommand{\appl}[2]{#1 ~ #2}
\newcommand{\lam}[3]{\lambda \tbind{#1}{#2}.~#3}
\newcommand{\lamx}[2]{\lam{x}{#1}{#2}}
\newcommand{\tbox}[2]{\texttt{box}~#1.~#2}
\newcommand{\tboxemp}[1]{\texttt{box}~#1}
\newcommand{\teval}[2]{\texttt{eval}~[#1]~#2}

\newcommand{\tunit}{()}
\newcommand{\mvar}[2]{{#1}\,\langle{#2}\rangle}

\newcommand{\compappl}[4]{\sappl{\sop{#1}{#2}}{#3}{#4}}
\newcommand{\contappl}[5]{\sappl{\sscont{#1}{#2}{#3}}{#4}{#5}}
\newcommand{\letbox}[3]{\texttt{let box}~#1~\texttt{=}~#2~\texttt{in}~#3}
\newcommand{\letboxu}[2]{\letbox{u}{#1}{#2}}
\newcommand{\happl}[6]{\sappl{\shandle{#1}{#2}{#3}{#4}}{#5}{#6}}
\newcommand{\ret}[1]{\texttt{ret}~#1}
\newcommand{\sappl}[3]{#2 \leftarrow #1;~#3}
\newcommand{\tfix}[6]{
  \texttt{let fix}~#1~(\tbind{#2}{#3})~\texttt{=}~\texttt{box}~#4.~#5~\texttt{in}~#6
}
\newcommand{\cifel}[3]{\texttt{if}~#1~\texttt{then}~#2~\texttt{else}~#3}

\newcommand{\sop}[2]{#1~#2}

\newcommand{\sscont}[3]{\mathtt{cont}~#1~#2~#3}
\newcommand{\shandle}[4]{\mathtt{handle}~#1~[ #2 ]~#3~#4}
\newcommand{\shandlesimp}[3]{\mathtt{handle}~#1~#2~#3}

\newcommand{\handbase}[3]{\textit{return}\,(#1, #2) \shortrightarrow #3}

\newcommand{\handid}[1]{\textit{id}_{#1}}
\newcommand{\handoparr}[5]{#1\,(#2, #3, #4) \shortrightarrow #5}
\newcommand{\hextend}[2]{#1, #2}
\newcommand{\heval}[1]{\llbracket #1 \rrbracket}

\newcommand{\hseqbase}[0]{\bullet}
\newcommand{\hseqclause}[4]{(#1, #2, #3.~#4)}
\newcommand{\hseqextend}[5]{#1, \hseqclause{#2}{#3}{#4}{#5}}

\newcommand{\hndl}[3]{{\langle #1 \rangle}^{#2}_{#3}}
\newcommand{\hcont}[3]{#1.#2.~#3}
\newcommand{\hndlseq}[2]{{\llparenthesis #1 \mid #2 \rrparenthesis}}
\newcommand{\subst}[3]{#1[#2/#3]}
\newcommand{\substc}[3]{\{#1/#2\}#3}
\newcommand{\substk}[5]{#1[\hcont{#2}{#3}{#4} \rotatedsim{} #5]}
\newcommand{\substkxy}[3]{\substk{#1}{x}{y}{#2}{#3}}
\newcommand{\substm}[4]{#1[#2. #3/\!/#4]}

\begin{document}

\lstset{basicstyle=\ttfamily}

\title{Contextual Modal Types for Algebraic Effects and Handlers}         
\subtitle{Extended Version}                     


\author{Nikita Zyuzin}
\affiliation{
  \position{PhD Student}
  \institution{IMDEA Software Institute}            
  \streetaddress{Campus de Montegancedo s/n}
  \city{Pozuelo de Alarcon}
  \state{Madrid}
  \postcode{28223}
  \country{Spain}                    
}
\email{nikita.zyuzin@imdea.org}          
\affiliation{
  \institution{Universidad Politécnica de Madrid}            
  \country{Spain}                    
}

\author{Aleksandar Nanevski}
\affiliation{
  \position{Associate Professor}
  \institution{IMDEA Software Institute}            
  \streetaddress{Campus de Montegancedo s/n}
  \city{Pozuelo de Alarcon}
  \state{Madrid}
  \postcode{28223}
  \country{Spain}                    
}
\email{aleks.nanevski@imdea.org}         

\begin{abstract}
  Programming languages with algebraic effects often track the
computations' effects using type-and-effect systems.
%
%
In this paper, we propose to view an algebraic effect theory of a
computation as a variable context; consequently,
we propose to track algebraic effects of a computation
with \emph{contextual modal types}. We develop \cmtte, a
novel calculus which tracks algebraic effects by a contextualized
variant of the modal $\Box$ (necessity) operator, that it inherits
from Contextual Modal Type Theory (CMTT).


Whereas type-and-effect systems add effect annotations on top of a
prior programming language, the effect annotations in \cmtte are
inherent to the language, as they are managed by programming
constructs corresponding to the logical introduction and elimination
forms for the $\Box$ modality. Thus, the type-and-effect system
of \cmtte is actually just a type system.

Our design obtains the
properties of local soundness and completeness, and determines the
operational semantics solely by $\beta$-reduction, as customary in
other logic-based calculi. In this view, effect handlers arise
naturally as a witness that one context (i.e., algebraic theory) can
be reached from another, generalizing explicit substitutions from
CMTT.


To the best of our knowledge, \cmtte is the first system to relate
algebraic effects to modal types.  We also see it as a step towards
providing a correspondence in the style of Curry and Howard that may
transfer a number of results from the fields of modal logic and modal
type theory to that of algebraic effects.

\end{abstract}



\maketitle

\section{Introduction}\label{sec:intro}
Languages with algebraic effects
\cite{plotkin2001,plotkin2002,plotkin2003} represent effects as calls
to operations from algebraic theories.  Effect handlers
\cite{pretnar2013} specify how these operations should be interpreted
when the computation using them executes.  Together, algebraic effects
and handlers provide a flexible way for composing effects in a
computation, and for evaluating a computation into a pure value by
handling all of its effects.

For example, to combine effects of state and exception in a
computation it suffices to merely use the operations from the
algebraic theory of state alongside the operations from the algebraic
theory of exceptions.  The task of interpreting the combination of
effects is delegated to the point when the computation executes
through a handler for the effects of both theories.  This is in
contrast to monads, where the semantics is provided upfront, and
remains unchanged through execution.

Typically, a type system for algebraic effects keeps track of the
operations of a program through effect annotations, deriving a
type-and-effect system.  A call to an operation is annotated with the
invoked effect, which propagates to the type of the enclosing
computation. Handling removes the effect from the
annotation. Crucially, the search for the best type-and-effect system
for algebraic effects is still open. For example, the recent work on
effect handlers in OCaml~\cite{sivaramakrishnan2021} highlights the
current lack of a design in which effect annotations combine well with
advanced language features such as polymorphism, modularity and
generativity.






In this paper, we propose an entirely type-based design for a language with
algebraic effects.
We consider the operations used by a
computation as its \emph{effect context}, and track it through the typing
mechanisms that we adapt from Contextual Modal Type Theory
(CMTT)~\cite{nanevski2008}.  In particular, we derive a novel calculus
for algebraic effects rooted in modal logic and modal type theory,
which we name Effectful Contextual Modal Type Theory, or \cmtte for
short.




CMTT derives from a contextual variant of intuitionistic modal
logic S4 through the annotation of the inference rules with proof terms and
a computational interpretation in the style of Curry-Howard correspondence.
Modal logics in general reason about truth in a universe of possible
worlds.  In the specific case of S4, the key feature is
the propositional constructor $\Box$ (called ``necessity'', or ``box''
for short).  The proposition $\Box{A}$ is considered proved in the
current world, if we can produce a proof of $A$ in every possible
future
world~\cite{Simpson94,benton98,alechina01categorical,pfenning2001}.
The derived computational interpretations resulted in type systems for
staged computation~\cite{Davies01jacm} and run-time code
generation~\cite{wickline98modal}.

\newcommand{\ctx}{\Gamma}
\newcommand{\bctx}{\Psi}

The contextual variant further indexes, or grades, the necessity
constructor with a context of propositions $\bctx$.  The proposition
$\boxt{\bctx}{A}$ is considered proved in the current world if we can
produce a proof of $A$ in every possible future world, using
\emph{only} the propositions from $\bctx$ as hypotheses. The modality
$\Box{A}$ is recovered when the context $\bctx$ is empty. The
computational interpretation transforms propositions to types, and
considers the type $\boxt{\bctx}{A}$ as classifying programs of type
$A$ that admit free variables from the variable context $\bctx$, and
no others.  The typing discipline has obtained calculi for
meta-programming with open code~\cite{Nanevski05jfp}, for dynamic
binding~\cite{nanevski03ppdp}, for typed tactics and proof
transformations in proof
assistants~\cite{veriml,stampoulisphd,pientka:flops10}, and has
explained the notion of meta variable and related optimizations in
logical frameworks~\cite{pientka:cade03}.

\cmtte applies the contextual discipline by taking the context $\Psi$
to be an algebraic theory---a signature of effect
operations\footnote{In this paper we consider only \emph{free}
  algebraic theories that contain operations, but no equations between
  operations.}---locally bound in a computation of type
$\boxt{\bctx}{A}$.
It allows \cmtte to integrate effects much stronger into the language than a
typical type-and-effect system would.
A typical type-and-effect system annotates the types of some existing
language without changing the terms.
In contrast, algebraic theories in \cmtte are manipulated at the level of
terms as well.
The constructs for context (algebraic theory in our view) binding and
instantiation logically correspond to introduction and elimination forms for the
modality.
We prove that these constructs are in
harmony~\cite{pfenning2009} by establishing their local soundness and
completeness. The latter properties ensure that the typing rules for
the modality are self-contained and independent of the other language
features, and in turn imply the desired modularity of the type system
design.



By connecting the disciplines of modal logic and of algebraic effects,
we see \cmtte as a step towards transferring the ideas between
them.  For example, the work on modal type theory has
considered polymorphism over contexts as first-class
objects~\cite{pientka:popl08,cave:lfmtp13}, which should have direct
equivalent in polymorphism over effect theories
\cite{hillerstrom2016,leijen2017,biernacki2019,zhang2019,brachthauser2020oopsla}.

Because of the foundation in modal logic, we also expect that \cmtte
will admit important logical properties and developments such as:
Kripke semantics and normalization by
evaluation~\cite{Simpson94,bak:imla17,ilik13,gratzer:icfp19}, proof
search via cut-free sequent calculus (corresponding computationally to
synthesizing programs with algebraic effects), clean equational
theory, proof of strong normalization, scaling to dependent
types~\cite{ahman2018, gratzer:icfp19}, etc.  We leave all of these
considerations for future work.

\subsection{Introducing Contextual Modality}

To see how contextual types apply to algebraic effects, consider the
algebraic theory of state with a single integer cell.  As in related
work on algebraic effects and handlers~\cite{bauer2018, bauer2015},
this theory consists of the operations $\mathit{get}: \mathit{unit}
\to \mathit{int}$ and $\mathit{set}: \mathit{int} \to \mathit{unit}$.
%
%
The following program increments the state by 1, and has type
$\mathit{int}$.
\[
\tbind{
  \sappl{\sop{\mathit{get}}{\tunit{}}}{x}{
    \sappl{\sop{\mathit{set}}{(x + 1)}}{\_}{
      \ret{x}
    }
  }
}{\intt{}}
\]
Intuitively, we model this theory as a \emph{variable context}
$\mathit{St} \eqdef \mathit{get}: \mathit{unit} \to \mathit{int},
\mathit{set}: \mathit{int} \to \mathit{unit}$,\footnote{As we shall
  see in Section~\ref{sec:overview}, the technical details of our
  variable typing will differ somewhat, and will make use of several
  different typing judgements, but the above types for $\mathit{get}$
  and $\mathit{set}$ in the context $\mathit{St}$ are sufficiently
  approximate for now.}  which is \emph{bound locally} in the scope of
the computation, and is also listed in the computation's type:
\begin{equation}
\mathit{incr} \eqdef{}
\tbind{
  \tbox{\mathit{St}}{
    \sappl{\sop{\mathit{get}}{\tunit{}}}{x}{
      \sappl{\sop{\mathit{set}}{(x + 1)}}{\_}{
        \ret{x}
      }
    }
  }
}{\boxt{\mathit{St}}{\intt{}}} \label{def:incr}
\end{equation}
The term constructor \texttt{box} is inherited from CMTT, and
\emph{simultaneously} binds all the variables from the supplied
context $\mathit{St}$. In this case, the effect operations
$\mathit{get}$ and $\mathit{set}$ from $\mathit{St}$ become available
for use within the scope of \texttt{box}.

Operationally, \texttt{box} \emph{thunks} the enclosed programs.
The boxed program doesn't execute until explicitly forced, and is considered
pure.
Semantically, the type $\boxt{\mathit{St}} A$ classifies computations that use
operations from the theory
$\mathit{St}$---\emph{but no other operations}---and return a value of
type $A$ upon termination. Of course, in general, we admit an
arbitrary variable context (resp.~algebraic theory) $\bctx$ to be
bound by \texttt{box} and $\boxt{\bctx}{A}$, not just the concrete one
$\mathit{St}$. Crucially, the \texttt{box} constructor is the
\emph{introduction form} for $\boxt{\bctx}{A}$.

\subsection{Context Reachability}

In a contextual type system one has to describe how the propositions
from one context $\bctx$, can be proved from the hypotheses in
another context $\bctx'$. When such proofs can be constructed, one
says that $\bctx$ is reachable from $\bctx'$, or that $\bctx'$ reaches
$\bctx$. As we discuss in Section~\ref{sec:background}, CMTT models
context reachability using explicit substitutions.
In \cmtte, we model context reachability somewhat differently, using
\emph{effect handlers}.

For example, the following is one possible handler for the algebraic
theory $\mathit{St}$.
\[
\begin{array}{r@{\ }c@{\ }l}
\mathit{handlerSt} & \eqdef & 
  (\!\!\!\begin{array}[t]{l}
      \handoparr{\textit{get}}{x}{k}{z}{
        \sscont{k}{z}{z}
      }, \\
      \handoparr{\textit{set}}{x}{k}{z}{
        \sscont{k}{\tunit{}}{x}
      }, \\
      \handbase{x}{z}{
        \ret{(x, z)}
      } )
    \end{array}
\end{array}
\]
As customary in languages for algebraic effects, an effect handler
consists of a number of clauses showing how to interpret each effect
operation upon its use in a computation.

In the case of
$\mathit{handlerSt}$, the clause for $\mathit{get}$ takes $x$ (of type
$\mathit{unit}$) as the call parameter, $k$ as the continuation at the
call site, and $z$ as the handler's current state at the call site.
It then calls $k$ with $z$ and $z$, thus modeling that $\mathit{get}$
returns the current value of the state (the first $z$), and proceeds
to execute handling from the unmodified state (the second $z$).  For
$\mathit{set}$, the call parameter $x$ has type $\mathit{int}$, and
the handler calls $k$ with $\tunit{}$ and $x$, thus modeling that
$\mathit{set}$ returns the (only) value $\tunit{}$ of $\unitt{}$ type,
and sets $x$ as the new state for the rest of the handling. The
\textit{return} clause applies when the handled computation returns a
value $x$ and ends in state $z$. In $\mathit{handlerSt}$, $x$ and $z$
are paired up and returned to the enclosing scope by $\ret{(x, z)}$.

This $\mathit{handlerSt}$ describes how to transform a program that
uses the theory $\mathit{St}$ into a program with no effects, and it
models state by passing it explicitly via continuation calls.  Because
its clauses use no operations, we can say that $\mathit{handlerSt}$ is
a witness for $\mathit{St}$ reaching the \emph{empty} algebraic theory
(i.e., the empty context).  Of course, it's perfectly possible for
handler clauses to invoke operations from a non-empty theory as
well. We shall see examples of this in Section~\ref{sec:overview}.

\subsection{Eliminating Contextual Modality}

Following CMTT, we adopt a $\mathtt{let}$-style constructor as an
elimination form for $\boxt{\bctx}{A}$. To illustrate 
this form, consider the following program that handles the
$\mathit{St}$ effects in $\mathit{incr}$ using
$\mathit{handlerSt}$.\footnote{In this paper, we don't consider
  handler variables. Thus, when we write \textit{handlerSt} in the
  code, the reader should assume that the definition of
  \textit{handlerSt} is simply spliced in. Similarly for algebraic
  theories, e.g., \textit{St}. Quantifying over contexts and
  substitutions has been done in the work on contextual
  types~\cite{cave:lfmtp13} and should directly transfer to our
  setting. We thus forego such addition and focus on the 
  fundamental connection between modalities and algebraic effects.}
\[
\letbox{u}{\textit{incr}}{\shandlesimp{u}{\textit{handlerSt}}{0}}
\]
The reduction of this program binds the variable $u$ to the term
inside \texttt{box} in the definition of \textit{incr} in
(\ref{def:incr}) and proceeds with the scope of $\texttt{let}$, where
the term is handled with $\textit{handleSt}$ starting from the initial
state $0$. This program evaluates to a pair of the produced value and state,
$(0, 1)$, as we will illustrate in detail in Section~\ref{sec:reductions}.

The reduction for $\texttt{let-box}$ is quite different from the one
associated with monads and monadic bind. Unlike with monads, the body
of $\textit{incr}$ isn't evaluated when \textit{incr} is bound to $u$,
but is evaluated later, when $u$ is handled. In particular, if $u$
itself doesn't appear in the scope of $\texttt{let-box}$, then the
computation bound to $u$ is never evaluated. 


This kind of \texttt{let} construct is associated with comonads, and
indeed the $\boxt{\bctx}$ modality of CMTT has been assigned a
comonadic semantics~\cite{gab+nan:japl13}. Similar \texttt{let}
constructs have been considered in other modal calculi for
(co)effects~\cite{nanevski03cmu,nanevski:phd04,gab:icfp16,orch:icfp19}.
Common to them is the reduction that provides an unevaluated
expression to a program term that's equipped for evaluating
it; in our case, to a handler matching the effect operations of the expression.
We don't consider the categorical semantics or the equational theory for our
variant of $\boxt{\bctx}{A}$ here.
However, one of the contributions of this paper is observing that the comonadic
elimination form applies to handling of algebraic effects as well.

\section{Review of Contextual Modal Type Theory}\label{sec:background}

The typing rules of CMTT rely on a two-context typing judgment from
Figure~\ref{fig:typing_cmtt}. The context $\Gamma$ contains the
variables of the ordinary lambda calculus $\tbind{x}{A}$, which we
refer to as \emph{value variables}, as they bind values in a
call-by-value semantics.
The context $\Delta$ contains
\emph{modal variables} $\mhypbind{u}{A}{\Psi}$ that are annotated with
a type $A$, but also with a context $\Psi$ of value variables.
Such a modal variable stands for a term that may depend on the (value)
variables in $\Psi$, but no other value variables.

The term $\tbox{\Psi}{e}$ captures this dependence. It binds $\Psi$ in
the scope of $e$, but also prevents $e$ from using any value variables
that may have been declared in the outside context. This is formalized
in the rule $\Box I$ where, reading the rule upside-down, the context
$\Gamma$ from the conclusion is removed from the premise. This means,
in particular, that the following term is well-typed\footnote{In
  examples, we freely use standard types such as $\intt{}$ and product
  types, integer constants, and functions such as $+$, $*$, pairing
  and projections, without declaring them in the syntax and the typing
  rules. They pose no formal difficulties. We also consider the
  application of such (pure) expressions to evaluate immediately.}
\[
\ejdgmt{\tbind{
    \tbox{\tbind{x}{\intt{}}, \tbind{y}{\intt{}}}{(x + y)}
  }{\boxt{\tbind{x}{\intt{}}, \tbind{y}{\intt{}}}{\intt{}}}}
\]
On the other hand, a term that, under $\texttt{box}$, uses value
variables other than those bound by $\texttt{box}$, cannot be ascribed
a type
\[
\enjdgmt{
  \lam{z}{\intt{}}{
    \tbox{\tbind{x}{\intt{}}, \tbind{y}{\intt{}}}{(x + y + z)}
  }
}
\]
In Section~\ref{sec:typing} we will use a similar rule for \cmtte to
capture that a term may contain effect operations from the algebraic
theory $\Psi$, but not from an outside algebraic theory $\Gamma$.

\begin{figure}
  \input{formal/typing_cmtt}
  \Description{Typing rules of CMTT.}
  \caption{Typing rules of CMTT.}\label{fig:typing_cmtt}
\end{figure}

However, \texttt{box} doesn't restrict the use of modal variables, as
the premise and the conclusion of the $\Box I$ rule share the $\Delta$
context.  To introduce a modal variable into $\Delta$ one uses the
\[
\letbox{u}{e_1}{e_2}\] construct (rule $\Box E$), which binds $u$ to
\emph{the unboxed body} of $e_1$ in the scope of $e_2$.  For example,
eliding the types and commas in iterated variable binding, if $e_1 =
\tbox{x\ y}{(x + y)}$, then $u$ is bound to $x + y$, which is a term
with free variables $x$ and $y$. In the scope of $e_2$, $u$ will be
declared in $\Delta$ as $\mhypbind{u}{\intt}{\tbind{x}{\intt},
  \tbind{y}{\intt}}$.

Because $u$ may be bound to a term with free variables from its
associated context $\Psi$, one can use $u$ in a program only after
providing a definition for \emph{all} of the variables in
$\Psi$. These definitions are given by explicit substitutions (rule
\textsc{esub}) that guard the occurrences of $u$ (rule \textsc{mvar}).
For example, the following term, in which $u$ is guarded by the
explicit substitution $(\esbind{5}{x},\esbind{2}{y})$, is well-typed
and evaluates to $7$.
\begin{equation}
\letbox{u}{\tbox{x\ y}{(x + y)}}{
    \mvar{u}{\esbind{5}{x}, \esbind{2}{y}} \label{unbox}
}
\end{equation}

Because \texttt{box} doesn't restrict the use of modal variables, the
following term, in which $u$ occurs within a \texttt{box}, is also
well-typed, with type $\Box{\intt}$, and it evaluates to
$\texttt{box}\,(5 + 2)$.\footnote{When binding the empty variable
  context $\ectxsym$, we abbreviate $\boxt{\ectxsym}A$ as
  $\Box{A}$ and $\tbox{\!\ectxsym}{e}$ as $\texttt{box}\,e$.}
\begin{equation}
\letbox{u}{\tbox{x\ y}{(x + y)}}{
    \texttt{box}\,{(\mvar{u}{\esbind{5}{x}, \esbind{2}{y}})} 
}\label{unboxbox}
\end{equation}

As mentioned in Section~\ref{sec:intro}, in contrast to monadic bind, binding $\tbox{\Psi}{e_1}$ to a modal variable
$u$ 
\[\letbox{u}{\tbox{\Psi}{e_1}}{e_2}\]
doesn't by itself cause the evaluation of $e_1$.  Whether $e_1$ is
evaluated depends on the occurrences of $u$ in $e_2$.
For instance, in the example term~(\ref{unboxbox}), $u$ is thunked
under a \texttt{box}. Thus, upon substituting $u$ with $x +
y$, the explicit substitution $(\esbind{5}{x}, \esbind{2}{y})$ executes to
produce $\texttt{box}\,(5+2)$, after which the evaluation stops (i.e.,
boxed terms are values). In contrast, in the example
term~(\ref{unbox}), $u$ is not thunked, thus the evaluation proceeds
for one more step to obtain $7$.  Moreover, $e_2$ may contain several
occurrences of $u$, some thunked and some not, and each guarded by a
different explicit substitution. For example, the following is a
well-typed term
\[
\begin{array}[t]{l}
  \texttt{let\ box}\ u = \tbox{x\ y}{(x + y)}\ \texttt{in}\\
   \qquad (\!\!\!\begin{array}[t]{l}
       \mvar{u}{\esbind{5}{x}, \esbind{2}{y}}, \\
       \tbox{x}{\mvar{u}{\esbind{3}{x}, \esbind{x}{y}} + \mvar{u}{\esbind{x^2}{x},  \esbind{\mvar{u}{\esbind{2x}{x}, \esbind{1}{y}}}{y}} + x}, \\
       \tbox{x\ y}{\mvar{u}{\esbind{y}{x}, \esbind{x}{y}}})
       \end{array}
\end{array}
\]
which evaluates to the triple
\[(7, \tbox{x}{(3 + x) + (x^2 + (2x + 1)) + x}, \tbox{x\ y}{y +
x})\]
of type $\intt \times \boxt{\tbind{x}{\intt}}{\intt} \times
\boxt{\tbind{x}\intt, \tbind{y}\intt}{\intt}$. 
%

In general, local soundness for $\Box$ in CMTT is formalized as the
$\beta$-reduction
\[
\letbox{u}{\tbox{\Psi}{e}}{e'}~\reduces{\beta}~\substm{e'}{\Psi}{e}{u}
\]
where $\substm{e'}{\Psi}{e}{u}$ is a \emph{modal substitution}, in
which $e$ substitutes for $u$ in $e'$, incurring the capture of
variables from $\Psi$. We elide its definition here (it can be found
in~\cite{nanevski2008}), but do emphasize its key property.  Namely,
upon modally substituting $e$ for $u$ in $\mvar{u}{\sigma}$, the
result is \emph{not} the expression $\mvar{e}{\sigma}$, but the
expression obtained after directly applying the substitution $\sigma$
to $e$.
%
%
We retain this design in \cmtte, where we develop a similar notion of
modal substitution, except with context $\Psi$ generalized to an
algebraic theory, and explicit substitutions replaced by effect
handlers. In particular, handling in \cmtte will apply to terms with
free \emph{operations}, which will let us incorporate handling into
$\beta$-reduction over open terms in Section~\ref{sec:reductions}.

We close the review of CMTT by noting that local completeness for
$\Box$ in CMTT is formalized as the following $\eta$-expansion
\[
e : \boxt{\Psi}{A}~\reduces{\eta}
\letbox{u}{e}{\tbox{\Psi}{\mvar{u}{\texttt{id}_\Psi}}}
\]
where $\texttt{id}_{\Psi}$ is the identity substitution for $\Psi$;
that is, if $\Psi = {x_1:A_1, \ldots, x_n:A_n}$, then
$\texttt{id}_\Psi \eqdef (\esbind{x_1}{x_1}, \ldots,
\esbind{x_n}{x_n})$. \cmtte will introduce a definition of an identity
effect handler and appropriately adapt the above $\eta$-expansion.

\section{Overview of \cmtte by Examples}\label{sec:overview}
To support algebraic effects and handlers, the syntax of \cmtte
(Figure~\ref{fig:lang_syntax}) diverges from CMTT in several important
aspects which we summarize below, before proceeding to illustrate
\cmtte through concrete programming examples.

\begin{enumerate}
\item We introduce a category of (effectful) computations. These are
  terms that explicitly sequence bindings of algebraic operations and
  continuation calls, and terminate with a return of a purely-functional
  result. Our formulation is based on the judgmental presentation of
  monadic computations by~\citet{pfenning2001}. The category of
  expressions is inherited from CMTT, and retains the
  purely-functional nature.
\item The context of value variables changes into effect context.  It
  now contains effect operations
  $\chypbind{\text{op}}{A}{B}$ and continuations
  $\khypbind{k}{A}{S}{B}$, each typed with a special new judgment.
  The modal type $\boxt{\Psi}{A}$ now classifies computations that use
  algebraic operations from the effect context $\Psi$, not
  purely-functional expressions as in CMTT.
\item The value variables move to the modal context. The
  intuition is that values, being purely-functional, can be regarded
  as computations in the empty effect theory. Thus, the typing
  $\tbind{x}{A}$ can be considered as a special case of
  the modal typing $\mhypbind{x}{A}{\ectxsym}$.
\item We replace explicit substitutions with handlers, and more
  generally, handling sequences, as we shall see. Similar to explicit
  substitutions, an effect handler binds a computation to each
  algebraic operation in a context $\Psi$. In each binding, a handler
  further provides access to the argument $x$ of the algebraic
  operation, and to the current continuation $k$, and to its current state $z$.
  A handler also provides a \textit{return} binding that applies when handling
  $\ret{e}$ terms.
\end{enumerate}

\begin{figure}
  \input{formal/language}
  \Description{Syntax of the \cmtte}
  \caption{The syntax of \cmtte. In effect contexts, $op$ ranges over 
    algebraic operations, and $k$ over continuations. $\Psi$ ranges 
    over effect contexts that are algebraic theories (i.e., that 
    contain only algebraic operations and no continuations). $P$
    ranges over base types, including at least the singleton type 
    $\unitt$ with value $\tunit$.\vspace{-4mm}}\label{fig:lang_syntax}
\end{figure}



\subsection{Effect Contexts, Algebraic Theories and Operations}
Similarly to CMTT, the \texttt{box} constructor in \cmtte binds all
the free variables in the underlying term. But in \cmtte, these free
variables are algebraic operations (henceforth, simply
\emph{operations}), and the term under the \texttt{box} is a
computation, not an expression. For example, we can box a computation
that reads and returns the state using the algebraic theory
$\textit{St} \eqdef{} \chypbind{\textit{get}}{\unitt{}}{\intt{}},
\chypbind{\textit{set}}{\intt{}}{\unitt{}}$ from the introduction, as
follows.\footnote{Except, this time, we type the algebraic operations
  in \textit{St} with the actual variable judgment from \cmtte.}
\[
\tbox{\textit{St}}{
  \sappl{\sop{\textit{get}}{\tunit{}}}{x}{\ret{x}}
}
\]
Here, we apply the operation \textit{get} on the argument of type
\unitt{}, record the result in $x$ and return $x$ right away.  The
typing judgement $\div$ for operations is new, and we use it to ensure
that an operation can only be invoked in a computation, not in an
expression, as expressions are purely-functional. We illustrate this
typing discipline in detail in Section~\ref{sec:typing}.

Importantly, \texttt{box} only binds contexts that contain operations
(i.e., variables typed by $\div$), even though the new \emph{effect
  context} $\Gamma$ in \cmtte can also hold continuation variables
(typed by $\kbinding$). We refer to the operation-only context as
\emph{algebraic theory} and we use $\Psi$ to range over such contexts.

In the above example, $\sop{\textit{get}}{()}$ is a \emph{statement},
with $\textit{get}$ an \emph{operation} and $()$ its
argument. Formally, statements are used with sequential composition,
analogous to monadic bind
\[\sappl{s}{x}{c}\]
which binds the result of $s$ to $x$ and continues with $c$. 
\cmtte features other forms of statements as well (e.g., continuation
application and handling of modal variables), which we discuss later
in this section. When the result of the statement is immediately
returned, we abbreviate the computation into the statement alone. For
example, we abbreviate the above term simply as
\[\tbox{\textit{St}}{\sop{\textit{get}}{\tunit{}}}\]

Of course, we retain from CMTT that \texttt{box} prevents the use of
outside operations. For example, the following term doesn't type check
because the inside $\texttt{box}$ declares the empty algebraic theory
as current, and rules out calls to the operations from the
outside theory $\mathit{St}$.
%
\[
\tbox{\textit{St}}{\texttt{box}\,({\sop{\textit{get}}{\tunit{}}})}
\]

\subsection{Modal Context, Value Variables and Handling of Modal Variables}\label{ssec:modal_context}
In \cmtte we move the value variables into the modal context
$\Delta$, reserving the context $\Gamma$ for operations and
continuations. This implements our intention that modal types in
\cmtte track effects, not the use of value variables as in
CMTT. Moving value variables into $\Delta$ has the additional
benefit that they now survive boxing.
For example, consider the computation \textit{incr} from the
introduction, but this time we want to increment the state not by $1$,
but by a user-provided integer $n$. We achieve this by
$\lambda$-binding $n$, and then invoking it under $\texttt{box}$.
\[
\mathit{incr}_n \eqdef \lam{n}{\intt{}}{
  \tbox{\mathit{St}}{
    \sappl{\sop{\mathit{get}}{\tunit{}}}{x}{
      \sappl{\sop{\mathit{set}}{(x + n)}}{y}{
        \ret{x}
      }
    }
  }
}
\]
The resulting term $\mathit{incr}_n$ is an expression (i.e., it's
purely functional), as it's a function whose body, being boxed, is
itself an expression.  Applying $\mathit{incr}_n$ to an integer
reveals the boxed computation. 
\[
\begin{array}[t]{l}
  \texttt{let\ box}\ u = \appl{\mathit{incr_n}}{2}\ \texttt{in}\
\shandlesimp{u}{\textit{handlerSt}}{0}
\end{array}
\]
As in CMTT, $\texttt{let-box}$ above eliminates the $\texttt{box}$,
binds the computation to a modal variable, and proceeds to handle the
incrementing computation with \textit{handlerSt} from the
introduction.\footnote{While the formal syntax for handling is
  $\shandle{u}{\Theta}{h}{e}$, by convention we elide $[\Theta]$ when
  $\Theta$ is empty, as it's here. We explain the need and the use of
  $\Theta$ in Section~\ref{ssec:handling_sequences}.}  Of course, we
could have used any other handler for the theory \textit{St}, just
like we could guard a modal variable $\mhypbind{u}{A}{\Psi}$ in CMTT
with any explicit substitution $\sigma$ that defines the variables
from $\Psi$.

An important term that becomes expressible by moving variables into
$\Delta$ is the following coercion of a value $x$ into a computation
in a given theory $\Psi$, commonly known as \emph{monadic unit}.
\[
\lambda x{:}A.\,\tbox{\Psi}{\ret{x}}
\]

\subsection{Handlers}
Effect handlers are similar to explicit substitutions in that they
replace operation variables with computations that define them. But
there are important differences as well. In particular, operation
clauses of a handler can manipulate the control flow of the target
program by providing access to the current (delimited) continuation of
the operation call. This facilitates a wide variety of
effects~\cite{pretnar2013,pretnar2015}, as we shall illustrate. 

\subsubsection{Return Clause}\label{sssec:return}
A handler in \cmtte always has a \textit{return} clause, which
applies when the handled computation terminates with $\ret{e}$.
The clause specifies how the handler processes $e$, together with the ending
handler state of the computation. To illustrate, consider the following
simple handler that consists only of a \textit{return} clause that
produces a pair of the final result and state:
\[
\textit{simple} \eqdef \handbase{x}{z}{\ret{(x, z)}}
\]
This handler handles no operations; thus its theory is empty, and it
only applies to computations over the empty theory. For example, we
can run $\textit{simple}$ on $\ret{42}$ with initial state $5$ as
follows.
\[
\letbox{u}{\texttt{box}\,(\ret{42})}{\shandlesimp{u}{simple}{5}}
\]
In the \textit{return} clause of \textit{simple}, the \texttt{ret} value of $42$
is bound to $x$ and $5$ is bound to $z$ to further return $(42, 5)$.
If we used a handler with a \textit{return} clause that makes no mention of
$x$ and $z$, e.g.,
\begin{equation}
\textit{simple} \eqdef \handbase{x}{z}{\ret{7}}\label{eq:handler_ret7}
\end{equation}
then the same program returns $7$ regardless of the initial state
supplied to \texttt{handle} for $u$.




\subsubsection{Updating Handler State}\label{sssec:handler_state}
Handler state is modified by invoking continuations in operation
clauses with new state values.
To illustrate, consider an extension of \textit{simple} with an
operation $\chypbind{\textit{op}}{\unitt{}}{\intt{}}$.
\begin{align*}
\mathit{simple^{*}} \eqdef{}
(&
\handoparr{\textit{op}}{x}{k}{z}{
  \sscont{k}{1}{(z + 4)},\\
&
\handbase{x}{z}{\ret{(x, z)}}
)}
\end{align*}
The handler clause for \textit{op} shows how $\mathit{simple}^{*}$
interprets occurrences of $\sop{\textit{op}}{e}$ in a handled term: it
binds the argument $e$ to $x$, the current continuation up to the enclosing
\texttt{box} to $k$, the current handler state to $z$, and proceeds with the
clause body.
In the clause body, the application $\sscont{k}{1}{(z + 4)}$ indicates
that $k$ continues with $1$ as the return value of \textit{op}, and
$z+4$ as the new state.
If we apply $\mathit{simple^{*}}$, with initial state $5$, to
$\ret{42}$, then handling returns $(42, 5)$, as before. If we apply
$\mathit{simple^{*}}$ with initial state $5$ to the computation that uses
\textit{op} non-trivially,
\begin{align}
&\letboxu{
  \tbox{
    \chypbind{\textit{op}}{\unitt{}}{\intt{}}
  }{
    \nonumber
    \\
    &
    \quad
    \sappl{\sop{\textit{op}}{\tunit}}{y_1}{
      \sappl{\sop{\textit{op}}{\tunit}}{y_2}{
        \sappl{\sop{\textit{op}}{\tunit}}{y_3}{
          \ret{(y_1 + y_2 + y_3)}
        }
      }
    }\label{eq:opopop}
  }
  \\
  &
}{
  \shandlesimp{u}{\mathit{simple^{*}}}{5}
  \nonumber
}
\end{align}
then handling returns $(3, 17)$.
To see this, consider that each call to \textit{op} is handled by
returning $1$, and therefore all $y_i$ hold $1$. Each call to
\textit{op} also increments the initial state by $4$. Thus, at the end
of handling, when the \textit{return} clause is invoked over
$\ret{(y_1 + y_2 + y_3)}$, $x$ and $z$ variables of the \textit{return} clause
will bind $3$ and $5 + 4 + 4 + 4 = 17$, respectively, to produce
$(3, 17)$. 

We can now revisit the handler for \textit{St} from the introduction.
\[
\begin{array}{r@{\ }c@{\ }l}
  \mathit{handlerSt} & \eqdef & 
  (\!\!\!\begin{array}[t]{l}
  \handoparr{\textit{get}}{x}{k}{z}{
    \sscont{k}{z}{z}
  }, \\
  \handoparr{\textit{set}}{x}{k}{z}{
    \sscont{k}{\tunit{}}{x}
  }, \\
  \handbase{x}{z}{
    \ret{(x, z)}
  } )
  \end{array}
\end{array}
\]
Because \textit{get} is handled by invoking $\sscont{k}{z}{z}$, it's
apparent that the handler interprets $\textit{get}$ as an operation
that returns the value of the current state $z$ and continues the
execution without changing this state. On the other hand,
$\sop{\textit{set}}{x}$ is handled by invoking
$\sscont{k}{\tunit{}}{x}$; thus, the handler interprets $\textit{set}$
as returning $()$ and changing the current state to $x$, as one would
expect.

\subsubsection{Parametrized Handlers and Function Purity}
As shown above, our handlers provide access to handler state, which is
shared between handler clauses, and can be updated through
continuations. In the parlance of algebraic effects, we thus provide
\emph{deep parametrized handlers}~\cite{hillerstrom2020}.
%
%
Most of the other algebraic effect systems, however, don't use
parametrized handlers.  For example, Eff~\cite{bauer2015} would encode
$\mathit{handlerSt}$ as
\[
\begin{array}{r@{\ }c@{\ }l}
  \mathit{handlerSt} & \eqdef & 
  (\!\!\!\begin{array}[t]{l}
  \textit{get}(x, k) \shortrightarrow {
    \lambda z.~ {
    \sscont{k}{z}{z}
    }
  }, \\
  \textit{set}(x, k) \shortrightarrow {
    \lambda z.~ {
    \sscont{k}{\tunit{}}{x}
    }
  }, \\
  \textit{return}(x) \shortrightarrow {
    \lambda z.~ {
    \ret{(x, z)}
    }
  } )
  \end{array}
\end{array}
\]
where the state $z$ is lambda-bound in each clause. 
Such a definition wouldn't type check in \cmtte as a continuation call
in the clause is a computation and can't be lambda-abstracted
directly, because a function body must be pure.  The latter, however,
is a common and important design aspect of languages that encapsulate
effects using types (e.g., Haskell).
Thus, parametrized handlers arise naturally as a way to avoid
lambda-abstraction over effectful handler clauses, if one is in a
setting where effects are encapsulated. The indirect benefit is
reducing closure allocation, the original motivation for parametrized
handlers~\cite{hillerstrom2020}.

\subsubsection{Uncalled Continuations}
If an operation clause of a handler doesn't invoke the continuation
$k$, this aborts the evaluation of the handled term.  For example, let
us extend $\mathit{simple^{*}}$ with an operation
$\chypbind{\textit{stop}}{\unitt{}}{\intt{}}$ that immediately returns
a pair of $42$ and the handler state.
\begin{align*}
\mathit{simple^{\dagger}} \eqdef{}
(&
\handoparr{\textit{op}}{x}{k}{z}{
  \sscont{k}{1}{(z + 4)}
},
\\
&
\handoparr{\textit{stop}}{x}{k}{z}{
  \ret{(42, z)},\\
&
\handbase{x}{z}{\ret{(x, z)}}
})
\end{align*}
If invoked on the computation from example~(\ref{eq:opopop}) with the same
initial state $5$, $\mathit{simple^{\dagger}}$ returns $(3, 17)$ as before, since
(\ref{eq:opopop}) doesn't make calls to \textit{stop}. If the handled
computation changed the middle call from \textit{op} to \textit{stop} as in
\begin{equation}
\sappl{\sop{\textit{op}}{\tunit}}{y_1}{
  \sappl{\sop{\textit{stop}}{\tunit}}{y_2}{
    \sappl{\sop{\textit{op}}{\tunit}}{y_3}{
      \ret{(y_1 + y_2 + y_3)}
    }
  }
}\label{eq:opstopop}
\end{equation}
then handling reaches neither the last \textit{op}, nor $\ret{(y_1 +
  y_2 + y_3)}$. It effectively terminates after handling
\textit{stop}, to return $(42, 9)$, as instructed by the
\textit{stop} clause of the handler. The state $9$ is obtained after
the initial state $5$ is incremented by $4$ through the handling of
the first call to \textit{op}.

Eliding continuations makes it possible to handle a theory of
exceptions~\cite{pretnar2015}
\[
\textit{Exn} \eqdef
\chypbind{\textit{raise}}{\unitt{}}{\bot}
\]
A handler for \textit{Exn}, such as the following one
\[
\begin{array}{r@{\ }c@{\ }l}
  \mathit{handlerExn} & \eqdef & 
  (\!\!\!\begin{array}[t]{l}
  \handoparr{\textit{raise}}{x}{k}{z}{
    \ret{42}
  }, \\
  \handbase{x}{z}{
    \ret{x}
  } )
  \end{array}
\end{array}
\]
can't invoke $k$ in the clause for $\textit{raise}$, because $k$ requires the
result of \textit{raise} as an input, which must be a value of the
uninhibited type $\bot$. As such an argument can't be provided, the handler
must terminate with some value when it encounters $\textit{raise}$, thus
precisely modeling how exceptions are handled in functional programming.

\subsubsection{Non-tail Continuation Calls}\label{sssec:counting_handler}

So far, our example handlers either invoked the continuation variable
$k$ as the last computation step (tail call), or did not invoke $k$ at
all. In \cmtte, $k$ may be invoked in other ways as well.

To illustrate, consider the following handler for the theory
$\chypbind{\textit{op},\textit{stop}}{\unitt{}}{\intt{}}$ from
before. The handler invokes $k$ in non-tail calls, to count the
occurrences of \textit{op} and \textit{stop} in the handled term.
\begin{align*}
\textit{handlerCount} \eqdef (&
\handoparr{\textit{op}}{x}{k}{z}{
  \sappl{\sscont{k}{1}{z}}{y}{
    \ret{(\pi_1\,y + 1, \pi_2\,y)}
  }
},
\\
&
\handoparr{\textit{stop}}{x}{k}{z}{
  \sappl{\sscont{k}{1}{z}}{y}{
    \ret{(\pi_1\,y, \pi_2\,y + 1)}
  }
},
\\
&
\handbase{x}{z}{\ret{(0, 0)}}
)
\end{align*}

To explain \textit{handlerCount}, first consider the \textit{return}
clause. Return clauses apply when handling values; as
values are pure, they can't call operations. Thus the \textit{return} clause
of \textit{handlerCount} returns $(0, 0)$ to signal that a
value contains $0$ occurrences of both \textit{op} and \textit{stop}.

Next consider the operation clauses. Each clause executes
$\sscont{k}{1}{z}$ to invoke the current continuation $k$---which
holds the handled variant of the remaining
computation---with argument $1$ and current state $z$.\footnote{Both
  $1$ and $z$ are irrelevant for the execution of the example, but
  they make the example typecheck.}  The obtained result, bound to
$y$, is a pair containing the number of uses for \textit{op} and \textit{stop} in the the
remaining computation. Each clause then adds 1 for the
currently-handled operation to the appropriate projection of $y$.
For example, handling~(\ref{eq:opstopop}) returns $(2, 1)$, as
expected.



\subsubsection{Multiply-called Continuations}
A handler clause in \cmtte may also invoke $k$ several times; in the literature
such continuations are usually called \emph{multi-shot}~\cite{bruggeman1996}.
This feature is useful for modeling the algebraic theory of
\emph{non-determinism}
\[
\mathit{NDet} \eqdef{}
\chypbind{\mathit{choice}}{\unitt{}}{\boolt{}}
\]
in which the operation \textit{choice} provides an unspecified Boolean value, as
in the following program.
\begin{equation}
\sappl{\sop{\textit{choice}}{\tunit{}}}{y}{
  \cifel{y}{
    \ret{4}
  }{
    \ret{5}
  }
}\label{eq:nondet}
\end{equation}
A handler for \textit{NDet} chooses how to interpret the
non-determinism. For example, the handler below enumerates all the
possible options for \textit{choice} and collects the respective
outputs of the target program into a list. Here $[x]$ is a list with
single element $x$, and ${++}$ is list append.
\[
\begin{array}{r@{\ }c@{\ }l}
  \mathit{handlerNDet} & \eqdef & 
  (\!\!\!\begin{array}[t]{l}
  \handoparr{\textit{choice}}{x}{k}{z}{
    \sappl{\sscont{k}{\textit{true}}{z}}{y_1}{
      \sappl{\sscont{k}{\textit{false}}{z}}{y_2}{
        \ret{(y_1\,{++}\,y_2)}
      }
    }
  }, \\
  \handbase{x}{z}{
    \ret{[x]}
  } )
  \end{array}
\end{array}
\]
Applying the handler to~(\ref{eq:nondet}) returns the list $[4, 5]$.

\subsection{Operations in Handler Clauses}\label{sec:explosive}
Our example handlers so far used no outside operations, and thus handled into
the empty theory.
But in \cmtte handlers can handle into other theories as well.
For example, the following is an alternative handler for \textit{St} which
throws an exception whenever the state is set to $13$; thus it uses the
operation \textit{raise} and handles \textit{St} into the theory \textit{Exn}.
\[
\begin{array}{r@{\ }c@{\ }l}
  \mathit{handlerExplosiveSt} & \eqdef & 
  (\!\!\!\begin{array}[t]{l}
  \handoparr{\textit{get}}{x}{k}{z}{
    \sscont{k}{z}{z}
  }, \\
  \handoparr{\textit{set}}{x}{k}{z}{
    \cifel{x = 13}{\sop{\textit{raise}}{\tunit{}}}{\sscont{k}{\tunit{}}{x}}
  }, \\
  \handbase{x}{z}{
    \ret{(x, z)}
  } )
  \end{array}
\end{array}
\]

Using \textit{handlerExplosiveSt} to handle $\appl{\mathit{incr}_n}{1}$ results
in a computation from the theory \textit{Exn}.
The type system in Section~\ref{sec:typing} will
require us to make the dependence on \textit{Exn} explicit. For
example, the following function \textit{explode} takes input $m$ and
builds a box thunk with~\textit{Exn}. Within it, we use
\textit{handlerExplosiveSt} over $u$ with $m$ as the initial
state, and return the first projection of the result.
\[
\begin{array}[t]{r@{\ }c@{\ }l}
\textit{explode} & \eqdef & \lambda m.\, \texttt{let\ box}\ u = \appl{\mathit{incr}_n}{1}\ \texttt{in}\\
& & \qquad \!
  \texttt{box}\ \textit{Exn}.\ \!\!\!\begin{array}[t]{l}
     x \leftarrow \shandlesimp{u}{\textit{handlerExplosiveSt}}{m};\\
     \ret{(\pi_1\, x)}
     \end{array}
\end{array}
\]

If we evaluate $\appl{\mathit{explode}}{0}$, we obtain
$\tbox{\textit{Exn}}{\ret{0}}$ as follows. The unboxed body of
$\appl{\mathit{incr}_n}{1}$ is first passed to
$\textit{handlerExplosiveSt}$ for handling with initial state
$0$. This results in the pair $(0, 1)$ indicating that
$\appl{\mathit{incr}_n}{1}$ read the initial state, and returned the
value read together with the incremented state. As we shall formalize
in Section~\ref{sec:reductions}, the binding of the pair to $x$ is
immediately reduced, and the first projection is taken, to obtain
$\ret{0}$, which is finally thunked with $\texttt{box}~\textit{Exn}$.

On the other hand, if we evaluate $\appl{\mathit{explode}}{12}$, then we obtain
$\tbox{\textit{Exn}}{\textit{raise}\,{()}}$, because the handling for the unboxed
body of $\appl{\mathit{incr}_n}{1}$ with $\textit{handlerExplosiveSt}$
terminates with \textit{raise} upon trying to set the state to 13.

We can proceed with handling by \textit{handlerExn} from initial state
$\tunit{}$, as in
\[
\letbox{v}{\appl{\mathit{explode}}{0}}{\shandlesimp{v}{\textit{handlerExn}}{()}}
\]
where the \textit{return} clause of \textit{handlerExn} returns $0$.
Or, if we change the initial state to $12$,
\begin{equation}
\letbox{v}{\appl{\mathit{explode}}{12}}{\shandlesimp{v}{\textit{handlerExn}}{()}}\label{ex:exploding}
\end{equation}
then the \textit{raise} clause of \textit{handlerExn} returns $42$.

\subsubsection{Identity Handler}\label{sec:identity}
Given a theory $\Psi$, the identity handler $\mathtt{id}_{\Psi}$
handles $\Psi$ into itself, or more generally, into any theory that
includes the operations of $\Psi$. Following the definition of
identity substitution of CMTT from Section~\ref{sec:background}, the
identity handler is not an \cmtte primitive, but a meta definition,
uniformly given for each $\Psi$.
\[
\begin{array}{r@{\ }c@{\ }l}
  \mathtt{id}_\Psi & \eqdef & 
  (\!\!\!\begin{array}[t]{l}
  \handoparr{\mathit{op}_i}{x}{k}{z}{
    \sappl{\sop{\mathit{op}_i}{x}}{y}{
      \sscont{k}{y}{z}
    }
  }, \\ 
  \handbase{x}{z}{
    \ret{x}
  } )
  \end{array}
\end{array}
\]
Each operation $\textit{op}_i$ of $\Psi$ is handled by invoking it
with the same argument $x$ with which $\textit{op}_i$ is encountered,
and passing the obtained result to the awaiting continuation. The
\textit{return} clause also just further returns the encountered value. The
handler
state $z$ is merely propagated by operation clauses, and thus doesn't
influence the computation. For convenience, we will thus consider that
in $\mathtt{id}_\Psi$ the variables $z$ are of type $\unitt$, and
always invoke $\mathtt{id}_\Psi$ with initial state $\tunit{}$.
Using the identity handler, we provide the \cmtte variant of $\eta$-expansion
for modal types in analogy to that of CMTT in Section~\ref{sec:background}.
\[
e : \boxt{\Psi}{A}~\reduces{\eta}
\letbox{u}{e}{\tbox{\Psi}{\shandlesimp{u}{\texttt{id}_\Psi}{\tunit{}}}}
\]

\subsubsection{Combining Theories}
Suppose that we are given the function \textit{safeDiv} that implements division
and raises an exception if the divisor is $0$.
\[
\textit{safeDiv} \eqdef \lam{x}{\intt{}}{
  \lam{y}{\intt{}}{
    \tbox{Exn}{
      \cifel{y = 0}{
        \sop{\textit{raise}}{\tunit{}}
      }{
        \ret{(x / y)}
      }
    }
  }
}
\]

Obviously, the body of \textit{safeDiv} abstracts over the
\textit{Exn} theory. Suppose that we now wanted to use
\textit{safeDiv} in a program that also operates over state, i.e.,
over theory \textit{St}. One way to do so would be to handle the calls
to \textit{safeDiv} by a handler, such as \textit{handlerExn} for
example, that catches the exception and returns a pure value.
But identity handlers provide another way as well. We can combine the
theories of \textit{St} and \textit{Exn} on the fly into a common
theory\footnote{We assume here that the combined theories don't share
  operation names, and can thus be concatenated without clashes and
  variable shadowing, as indeed is the case for \textit{St} and
  \textit{Exn}.
}, and then simply delay resolving
exceptions to the common theory.
\[
\textit{divFromState} \eqdef{} 
   \mathtt{box}\ \textit{St}, \textit{Exn}.\!\!\begin{array}[t]{l}
      y \leftarrow {\sop{\textit{get}}{\tunit{}}};\\
      \texttt{let\ box}\ u = \appl{\appl{\textit{safeDiv}}{42}}{y}\ \texttt{in}\\
      \shandlesimp{u}{\mathtt{id}_{\textit{Exn}}}{\tunit{}}
    \end{array}
\]
For example, \textit{divFromState} reads the current state into $y$
and invokes \textit{safeDiv} on 42 and $y$, re-raising an eventual
exception by handling with $\texttt{id}_{\textit{Exn}}$. The program
\textit{divFromState} explicitly boxes over both \textit{St} and
\textit{Exn} to signal that the theories are combined into a common
context.

\subsection{Handling Sequences}\label{ssec:handling_sequences}




In \cmtte, a handler is syntactically always applied to a variable,
rather than to an expression, similarly to explicit substitutions in
CMTT. The design allows us to express handling as part of modal
variable substitution and thus, correspondingly, as part of
$\beta$-reduction. Unlike in other calculi where handling is applied
over closed terms during evaluation, our handling applies over terms
with free variables, just like $\beta$-reduction. Moreover,
$\beta$-reduction suffices to implement handling. Indeed, if we
allowed applying handlers to general computations, as in
\[
\shandlesimp{c}{h}{s}
\]
then $\shandlesimp{c}{h}{s}$ must be a redex, with reductions for it
additional to $\beta$-reduction.

However, this design causes a problem with handler composition, which
occurs in the presence of free variables.
To see the issue, consider the following example which slightly
reformulates~(\ref{ex:exploding}).
\begin{equation}
\begin{array}{l}
\texttt{let\ box}\ v = \texttt{box}\,\textit{Exn}.\ (\!\!\!\begin{array}[t]{l}
      \texttt{let\ box}\ u = \appl{\mathit{incr}_n}{1}\ \texttt{in}\\
      \qquad \!\!\!\begin{array}[t]{l}
             x \leftarrow \shandlesimp{u}{\textit{handlerExplosiveSt}}{12};\\
             \ret{(\pi_1\, x)})
             \end{array} 
      \end{array}\\
\texttt{in}\ \shandlesimp{v}{\textit{handlerExn}}{\tunit{}}
\end{array}
\label{eq:explode_composition}
\end{equation}
This term reduces by first binding $v$ to the whole computation under
the \textit{Exn} box. Next, the computation $v$ is handled by
\textit{handlerExn}, before any reduction in $v$ itself. Therefore we
must be able to handle all the subterms of $v$, including
\[\shandlesimp{u}{\textit{handlerExplosiveSt}}{12}\]
which already contains handling by \textit{handlerExplosiveSt}. To
handle this expression, we can't simply return
$\shandlesimp{(\shandlesimp{u}{\textit{handlerExplosiveSt}}{12})}{\textit{hadnlerExn}}{\tunit{}}$. That
isn't syntactically well-formed, as it applies a handler to a general
computation, rather than a variable. Intuitively, we would like to
return an expression of the form
\begin{equation}
\shandlesimp{u}{(
\textit{handlerExn}
\circ
\textit{handlerExplosiveSt})}
{12}
\end{equation}
which records that eventual substitutions of $u$ must first execute
\textit{handlerExplosiveSt} over the substituted term, followed by
\textit{handlerExn}.

One may think that composing \textit{handlerExn} and
\textit{handlerExplosiveSt} requires simply applying
\textit{handlerExn} to all the clauses of \textit{handlerExplosiveSt};
unfortunately, this isn't correct.\footnote{Incidentally, such direct
  approach works in CMTT for composing explicit
  substitutions~\cite{nanevski2008}, but handlers are much more
  involved than explicit substitutions.}  When applying
\textit{handlerExplosiveSt} and \textit{handlerExn} in succession to
some computation $c$, the handling by \textit{handlerExn} requires
instantiating its continuation variables with computations obtained
after \textit{handlerExplosiveSt} is applied to $c$. In the above
case, $c$ is unknown, with variable $u$ serving as a
placeholder. Until $u$ is substituted, we don't have access to
continuations necessary to execute $\textit{handlerExn}$, and thus
can't eagerly apply it to the clauses of
$\textit{handlerExplosiveSt}$.

To encode the composition of handlers, we thus introduce the
\emph{handling sequence} $\Theta$. The term
\[
\shandle{u}{\Theta}{h}{e}
\]
stands for applications of handlers from the list $\Theta$,
in sequence from left to right, on a term bound to $u$, followed by
the handler $h$ running on the result, using initial state $e$.
However, $\Theta$ can't merely be a list of handlers, as each handler
must also be associated with the initial state, and the continuation
that's appropriate for it. For example, the immediate
$\beta$-reduction of example~(\ref{eq:explode_composition}) will be
\begin{align*}
&
\letbox{u}{\appl{\mathit{incr}_n}{1}}{
  \\
  &
\qquad  \shandle{u}{\hseqclause{handlerExplosiveSt}{12}{x}{\ret{(\pi_1~x)}}}{
    \textit{handlerExn}
  }{\tunit{}}
}
\end{align*}
recording that $\ret{(\pi_1~x)}$ is the immediate continuation of
$\shandlesimp{u}{\textit{handlerExplosiveSt}}{12}$
in~(\ref{eq:explode_composition}).

Finally, the need for handler sequences arises only when expressing
the intermediate results of program reduction, and we don't expect
that a user of our system will ever write them by hand. The simplified
syntax without $\Theta$ suffices for source code examples.

\section{Typing}\label{sec:typing}
\cmtte uses the following variable judgements:
\begin{align*}
&\tbind{x}{A} \in \Delta && \text{values of type $A$}
\\
&\mhypbind{u}{A}{\Psi} \in \Delta && \text{computations of type $A$ in an algebraic
theory $\Psi$}
\\
&\chypbind{op}{A}{B} \in \Gamma && \text{operations with input type $A$ and output
  type $B$}
\\
&\khypbind{k}{A}{S}{B} \in \Gamma && \text{continuations with input type $A$, state type 
  $S$, and output type $B$}
\end{align*}

We also have a separate judgement for each of the syntactic categories:
\begin{align*}
&\pjdgmt{\Delta}{e}{A} && \text{$e$ is an expression of type $A$}
\\
&\jdgmt{\Delta}{\Gamma}{\cbind{c}{A}} && \text{$c$ is a computation of type $A$}
\\
&\jdgmt{\Delta}{\Gamma}{\sbind{s}{A}} && \text{$s$ is a statement of type $A$}
\\
&\jdgmt{\Delta}{\Gamma}{\hbind{h}{A}{\Psi}{S}{B}} && \text{$h$ is a handler for
  computations of type $A$ in an algebraic theory $\Psi$ that}
  \\
  & && \text{uses state parameter of type $S$ and handles into computations of
    type $B$}
\\
&\jdgmt{\Delta}{\Psi}{\hseqbind{\Theta}{A}{\Psi'}{B}} && \text{$\Theta$ is a
  handling sequence for computations of type $A$ in an algebraic}
  \\
  & && \text{theory $\Psi'$ that produces computations of type $B$ in a theory $\Psi$}
\end{align*}

\begin{figure}
  \input{formal/typing_lang}
  \vspace{-2mm}
  \Description{Typing rules of \cmtte.}
  \caption{Typing rules of \cmtte.}\label{fig:typing}
  \vspace{-4mm}
\end{figure}

We show the typing rules of \cmtte on Figure~\ref{fig:typing}.
Let us review their key parts:
\begin{enumerate}
\item The variable typing rules \textsc{var}, \textsc{op},
  \textsc{cont} and \textsc{mvar} are grouped first.  The rule
  \textsc{mvar} guards a modal variable $u$ with a handling sequence
  $\Theta$, handler $h$, and initial state $e$. It ensures that the
  context $\Psi$ of $u$ matches the range theory of $\Theta$, and that
  the context $\Psi'$ of $\Theta$ matches the range theory of
  $h$. Thus, the handling of $u$ proceeds by $\Theta$, then $h$.
  The judgment for handling sequences formalizes this intermediary role of
  $\Theta$ by using theory $\Psi$ that contains only operations for its effect
  context, as opposed to the more general effect context $\Gamma$ that also
  includes continuations.

\item Unlike in CMTT, the judgment for expressions elides the context
  $\Gamma$, because expressions can't have effects. On the other hand,
  the rules dealing with function and modal types are almost
  unchanged.  The only distinctions are that $\Box{I}$ now boxes
  computations instead of expressions (introducing an effect context
  $\Psi$, not replacing it), and that we also have a rule
  $\Box{E}$-\textsc{comp} that eliminates modal types into
  computations.  The latter is a standard feature of
  \texttt{let}-forms in calculi with multiple
  judgments~\cite{pfenning2001}.

\item The \textsc{ret} rule coerces pure expressions into
  computations, while the \textsc{bind} rule sequentially composes a
  statement $s$ with a computation $c$.
  Notice that the term $\ret{e}$
  is a computation, not a statement. Thus, strictly speaking, the
  syntax $x \leftarrow \ret{e}; c$ is not valid in \cmtte. This is not
  a restriction in programming, as the intended behavior can be
  introduced as a syntactic sugar via boxing.\footnote{One solution is 
  $x \leftarrow \ret{e}; c \eqdef 
     \letbox{u}{\texttt{box}\,{(\ret{e})}}{x \leftarrow
        {\shandlesimp{u}{\texttt{id}_\ectxsym}{\tunit{}}}; c}$, and there are
     others.}

\item The rules \textsc{reth} and \textsc{oph} deal with building
  handlers.  \textsc{reth} is the base case, giving a handler for the
  empty theory $\ectxsym$ which thus contains only the \textit{return}
  clause. \textsc{oph} is the inductive case, adding a clause for a
  new operation \textit{op} to a handler $h$, thus extending $h$'s
  theory $\Psi$ with $\chypbind{op}{A}{B}$. We add new clauses to the
  right, but retain the syntax from Section~\ref{sec:overview},
  where the \textit{return} clause is presented as the last in a handler.

\item The rules \textsc{emph} and \textsc{seqh} deal with building
  handling sequences. \textsc{emph} gives an empty sequence, and
  \textsc{seqh} adds a handler clause to $\Theta$, ensuring proper
  sequencing, i.e., that the theory $\Psi'$ of the added handler $h$
  matches the effect context of $\Theta$.
\end{enumerate}

\subsection{Examples}
We illustrate the typing rules of \cmtte by showing a few example
typing judgments and derivations. We assume the theories of state and
exceptions from Section~\ref{sec:overview}.
\begin{align*}
\textit{St} & \eqdef \chypbind{\textit{get}}{\unitt}{\intt{}}, \chypbind{\textit{set}}{\intt{}}{\unitt{}}\\
\textit{Exn} & \eqdef \chypbind{\textit{raise}}{\unitt{}}{\bot}
\end{align*}

\begin{enumerate}
\item $\vdash \tbox{\textit{St}}{\sop{\textit{get}}{{\tunit{}}}} : \boxt{\textit{St}}{\intt{}}$

\item $\vdash 
  \tbind{
    \lambda {n}.\,{
      \tbox{\mathit{St}}{(
        \sappl{\sop{\mathit{get}}{\tunit{}}}{x}{
          \sappl{\sop{\mathit{set}}{(x + n)}}{y}{
            \ret{x}
          }
        })
      }
    }
  } {\funt{\intt{}}{\boxt{\textit{St}}{\intt{}}}}$

\item $\not\vdash \tbox{\textit{St}}{\ret{(\texttt{box}\,({\sop{\textit{get}}{\tunit{}}}))}}$

\item $\vdash \textit{handlerSt} \div \nhandlertype{\intt{}}{\textit{St}}{\intt{}}{\intt{} \times \intt{}}$

\item $\vdash \hbind{\textit{handlerExn}}{\intt{}}{\textit{Exn}}{\unitt{}}{\intt{}}$

\item $\textit{Exn} \vdash
  \hbind{
    \textit{handlerExplosiveSt}
  }{\intt{}}{\textit{St}}{\intt{}}{\intt{} \times \intt{}}$

\item $\textit{Exn} \vdash
  \hseqbind{
    \hseqclause{handlerExplosiveSt}{12}{x}{\ret{(\pi_1~x)}}
  }{\intt{}}{\textit{St}}{\intt{}}$

\item $\mhypbind{u}{\intt{}}{\textit{St}}; \ectxsym \vdash
    \shandle{u}{\hseqclause{handlerExplosiveSt}{12}{x}{\ret{(\pi_1~x)}}}
               {\textit{handlerExn}}{\tunit{}}
     \div \intt{}$
\item If $\Psi \subseteq \Psi'$, then $\Psi' \vdash \texttt{id}_\Psi \div \nhandlertype{A}{\Psi}{\unitt{}}{A}$

\item If $\Psi \subseteq \Psi'$, then $\vdash \lambda e.\,\letbox{u}{e}{\tbox{\Psi'}{\shandlesimp{u}{\texttt{id}_\Psi}{\tunit{}}}} : \boxt{\Psi}{A} \rightarrow\boxt{\Psi'}{A}$

\item $\vdash \lambda x.\,{\tbox{\Psi}{\ret{x}}} : \funt{A}{\boxt{\Psi}{A}}$

\item $\vdash \lambda f.\,\lambda x.\,\!\!\!\begin{array}[t]{l}
       \texttt{let\ box}\ u = f\ \\
       \texttt{\hphantom{let}\ box}\ v = x\\
       \texttt{in}\ \texttt{box}\ \Psi.\ \!\!\!\begin{array}[t]{l}
           a \leftarrow {\shandlesimp{u}{\texttt{id}_\Psi}{\tunit{}}};\\
           b \leftarrow {\shandlesimp{v}{\texttt{id}_\Psi}{\tunit{}}};\\ 
           \ret{(\appl{a}{b})} : \boxt{\Psi}{(A \rightarrow B)} \rightarrow \boxt{\Psi}{A} \rightarrow \boxt{\Psi}{B}
         \end{array}
      \end{array}$


\item 
$\vdash \lambda x.\, \!\!\!\begin{array}[t]{l}
      \texttt{let\ box}\ u = x\ \texttt{in}\\ 
      \texttt{box}\ \Psi.\, \!\!\!\begin{array}[t]{l}
          a \leftarrow \shandlesimp{u}{\texttt{id}_{\Psi}}{\tunit{}};\\
          \texttt{let\ box}\ v = a\ \texttt{in}\
            \shandlesimp{v}{\texttt{id}_{\Psi}}{\tunit{}}  : \boxt{\Psi}{\boxt{\Psi}{A}} \rightarrow \boxt{\Psi}A 
       \end{array}
    \end{array}$

\end{enumerate}

\subsubsection{Typing Computations}
Derivation (1) shows that we can invoke an operation declared in the
current theory. The derivation expands the syntactic sugar whereby
the computation consisting of a statement $s$ actually abbreviates $x
\leftarrow s; \ret{x}$.
{\small
\begin{mathpar}
\inferrule*[right=$\Box I$]
 {
  \inferrule*[right=bind]
   {
     \inferrule*[right=op]
     {
       \textit{get} \div \unitt\Rightarrow \intt{} \in \textit{St}\\
       \inferrule*{\hbox{}}{\textit{St} \vdash \tunit{} : \unitt}
     }
     {
       \textit{St} \vdash \sop{\textit{get}}{\tunit{}} \div_s \intt{}
     }
     \\
     \inferrule*[right=ret]
     {
       \inferrule*[right=var]
       {\hbox{}}
       {x : \intt{} \vdash x : \intt{}}
     }
     {
       x : \intt{}; \textit{St} \vdash \ret{x} \div \intt{}
     }
  }
  {
   \textit{St} \vdash x \leftarrow \sop{\textit{get}}{\tunit{}}; \ret{x} \div \intt{}
  }
}
{
\vdash \tbox{\textit{St}}{\sop{\textit{get}}{\tunit{}}} : \boxt{\textit{St}}{\intt{}}
}
\end{mathpar}}

On the other hand, we can try to derive (3) as follows
{\small
\begin{mathpar}
\inferrule*[right=$\Box I$]{
  \inferrule*[right=ret]{
    \inferrule*[right=$\Box I$]{
      \not\vdash \sop{\textit{get}}{\tunit{}} \div \intt
    }{
      \textit{St} \not\vdash {\texttt{box}\,({\sop{\textit{get}}{\tunit{}}})} : {\boxt{\ectxsym}{\intt}}
    }
  }{
    \textit{St} \not\vdash \ret{(\texttt{box}\,({\sop{\textit{get}}{\tunit{}}}))} \div {\boxt{\ectxsym}{\intt}}
  }
}{
\not\vdash \tbox{\textit{St}}{\ret{(\texttt{box}\,({\sop{\textit{get}}{\tunit{}}}))}} : \boxt{\textit{St}}{\boxt{\ectxsym}{\intt}}
}
\end{mathpar}}
but the derivation tree cannot be completed because \textit{get} doesn't appear in the effect context available at the top.

\subsubsection{Typing Handlers}
Next consider the derivations (4-6).  In the handler typing judgment,
the declared components $\nhandlertype{A}{\Psi}{S}{B}$ should be
understood as an input that guides the rest of derivation, and
determines the types of the variables bound in operation clauses.
For example, let us show why (4) holds.
Recall the definition of \textit{handlerSt}.
\[
 \vdash (\!\!\!\begin{array}[t]{l}
  \handoparr{\textit{get}}{x}{k}{z}{
    \sscont{k}{z}{z}
  }, \\
  \handoparr{\textit{set}}{x}{k}{z}{
    \sscont{k}{\tunit{}}{x}
  }, \\
  \handbase{x}{z}{
    \ret{(x, z)}
  } ) \div \raisebox{0pt}[0pt][0pt]{$\nhandlertype{\intt{}}{\textit{St}}{\intt{}}{\intt{} \times \intt{}}$}
  \end{array}
\]
in the typing of the \textit{return} clause we have to assume $x:\intt$,
as $\intt{}$ is the declared input type of the handler, and $z :
\intt{}$, as $\intt{}$ is the declared type of handler state. The
clause returns a pair of type $\intt{}\times\intt{}$, which matches
the result type of the handler, by rule \textsc{reth}.

Similarly, in the typing of the \textit{get} clause, we must assume
\begin{itemize}
\item $x : \unitt$, as $\unitt$ is the input type of \textit{get} in \textit{St}
\item $\khypbind{k}{\intt{}}{\intt{}}{\intt{}\times\intt{}}$. The
  first $\intt{}$ is the result type of \textit{get} in
  \textit{St}. The second $\intt{}$ is the declared type of handler
  state. The result type of $k$ matches the declared result type of
  the handler.
\item $z : \intt$ because $\intt$ is the declared type of handler
  state.
\end{itemize}
Under these assumptions, $\sscont{k}{z}{z}$ has type
$\intt{}\times\intt$ by rule \textsc{cont}, matching the result type
of the handler by rule \textsc{oph}.

\subsubsection{Typing Handling Sequences}

The starting point for derivation (7) is the judgment (6) of
\textit{handlerExplosiveSt}.  
In sequence (7), we provide $12$ as the initial state, thus
matching the type $\intt{}$ of handler state of
\textit{handlerExplosiveSt}. We also provide continuation
$x.\,\ret{(\pi_1~x)}$, which modifies the output type
$\intt{}\times\intt{}$ of \textit{handlerExplosiveSt} into the output
type $\intt{}$ of the sequence (7).

This modification of the output type of (7) into $\intt{}$ is
essential to typecheck (8). By the typing rule \textsc{oph}, the
output type of a sequence must match the input type of the extending
handler, and in our case, $\intt{}$ is \textit{handlerExn} input type
by derivation (5).

\subsubsection{Generic Modal Typing Derivations}
Derivation (9) is easy to establish by induction on $\Psi$.
Then derivation (10) shows that we can coerce a computation from
theory $\Psi$ into a larger theory $\Psi'$. In the special case when
$\Psi' = \Psi$, the derivation establishes the validity of
$\eta$-expansion for modal types, i.e., local soundness. We show
derivation (10) in some detail below, assuming derivation (9) and
weakening in both contexts. We also elide the first steps
involving $\lambda$, as they are standard.
{\small
\begin{mathpar}
\inferrule*[right=$\Box E$]{
  \inferrule*{}{
    \pjdgmt{\Delta}{e}{\boxt{\Psi}{A}}
  }
  \\
  \inferrule*[right=$\Box I$]{
    \inferrule*[right=mvar]{
      \inferrule*{}{
        \jdgmt{\Delta, \mhypbind{u}{A}{\Psi}}{\Psi'}{
          \hbind{
            \texttt{id}_\Psi
          }{A}{\Psi}{\unitt{}}{A}
        }
      }
      \\
      \inferrule*{ }{
        \pjdgmt{\Delta, \mhypbind{u}{A}{\Psi}}{
          \tunit{}
        }{\unitt{}}
      }
    }{
      \jdgmt{\Delta, \mhypbind{u}{A}{\Psi}}{\Psi'}{
        \cbind{
          \shandlesimp{u}{\texttt{id}_\Psi}{\tunit{}}
        }{A}
      }
    }
  }{
    \pjdgmt{\Delta, \mhypbind{u}{A}{\Psi}}{
      \tbox{\Psi'}{\shandlesimp{u}{\texttt{id}_\Psi}{\tunit{}}}
    }{\boxt{\Psi'}{A}}
  }
}{
  \pjdgmt{\Delta}{
    \letbox{u}{e}{\tbox{\Psi'}{\shandlesimp{u}{\texttt{id}_\Psi}{\tunit{}}}}
  }{\boxt{\Psi'}{A}}
}
\end{mathpar}
}

Derivations (11-13) show that the modal type $\boxt{\Psi}{A}$
satisfies the usual functions required of monads in
programming. Derivation (11) shows that a value of type $A$ can be
coerced into a computation of type $\boxt{\Psi}{A}$.
Derivation (12) shows that a function application is stable under
algebraic theory; that is, if a function and the argument share the
algebraic theory, so will the result. This is property $K$ of modal
logic.
Finally, derivation (13) shows that iterated modalities can be
coalesced. This is property $C4$ of modal logic.


\section{Reductions and subsidiary operations}\label{sec:reductions}
As in CMTT, $\beta$-reduction of modal types is the key notion in
\cmtte. This, together with $\beta$-reduction on function types, is the
basis for the \cmtte operational semantics (Figure~\ref{fig:op_sem}).
\begin{align*}
  \letbox{u}{\tbox{\Psi}{c}}{c'} ~ &\reduces{\beta} ~ \substm{c'}{\Psi}{c}{u}
  \\
  \letbox{u}{\tbox{\Psi}{c}}{e} ~ &\reduces{\beta} ~ \substm{e}{\Psi}{c}{u}
\end{align*}
The reduction depends on \emph{modal substitutions}
$\substm{c'}{\Psi}{c}{u}$ (resp.~$\substm{e}{\Psi}{c}{u}$), which
substitute computation $c$ for a modal variable $u$ in a computation
$c'$ (resp.~expression $e$), with $\Psi$ indicating the operators
bound in $c$.
We define modal substitution in the
Appendix~\ref{sec:sub_ctx_var}, and in this section
illustrate it on an example, along with a number of subsidiary
operations (Figure~\ref{fig:auxdef}) that modal substitution invokes.
These operations are:
\begin{itemize}
\item \emph{Monadic substitution} $\substc{c'}{x}{c}$ sequentially
  composes $c'$ before $c$, using variable $x$ as the connection. The
  definition follows closely a similar notion
  from~\citet{pfenning2001}. 
\item \emph{Continuation substitution} $\substkxy{c}{c'}{k}$ replaces
  the continuation variable $k$ in $c$ with a computation $c'$ that binds
  variables $x$ and $y$.
\item \emph{Handling} $\hndl{c}{h}{e}$ applies the handler
  $h$ over computation $c$, using $e$ as the current state.
\item \emph{Handling sequencing} $\hndlseq{c}{\Theta}$ applies the
  handling sequence $\Theta$ to the computation $c$.
\item \emph{Expression substitutions} $\subst{c}{e}{x}$ and
  $\subst{e'}{e}{x}$ replace an expression $e$ for variable $x$ into
  computation $c$ and expression $e'$ respectively. The last two are
  standard notions, so we use them without an explicit definition.
\end{itemize}

\begin{figure}
\begin{subfigure}{\textwidth}
  \input{formal/compsub}
  \Description{Monadic substitution.}
  \caption{Monadic substitution.}\label{fig:substc_def}
\end{subfigure}

\begin{subfigure}{\textwidth}
  \input{formal/contsub}
  \Description{Continuation substitution.}
  \caption{Continuation substitution.}\label{fig:substk_def}
\end{subfigure}

\begin{subfigure}{\textwidth}
  \input{formal/handling}
  \Description{Handling.}
  \caption{Handling.}\label{fig:hndl_def}
\end{subfigure}

\begin{subfigure}{\textwidth}
  \input{formal/handling_seq}
  \Description{Handling sequencing.}
  \caption{Handling sequencing.}\label{fig:hndl_seq_def}
\end{subfigure}

\Description{Definition of auxiliary operations.}
\caption{Definition of operations auxiliary to modal substitution. The
  definitions are recursive but well-founded (hence, also terminating)
  as each operation either makes recursive calls involving strictly
  smaller subterms, or invokes an operation that has been defined
  ahead of it in the figure. As conventional, we assume that all the
  bound variables are $\alpha$-renamed to avoid capture by the
  operations.}\label{fig:auxdef}
\vspace{-4mm}
\end{figure}

\subsection{Example}
We illustrate all these notions by tracing the reduction steps of the
following example. We present the steps on
Figure~\ref{fig:incr_reduction} and we also comment on them below.
\[
\letbox{u}{\appl{\mathit{incr_n}}{1}}{
  \shandlesimp{u}{\textit{handlerSt}}{0}}\]

\begin{figure}
  \input{formal/incr_reduction}
  \Description{Process of reduction for $\letbox{u}{\appl{\mathit{incr_n}}{1}}{
      \shandlesimp{u}{\textit{handlerSt}}{0}
    }$}
  \caption{Process of reduction for $\letbox{u}{\appl{\mathit{incr_n}}{1}}{
      \shandlesimp{u}{\textit{handlerSt}}{0}
    }$.}\label{fig:incr_reduction}
\end{figure}

In the first step~(\ref{eq:incr_reduction_1}), we unfold the
definition of $\appl{\mathit{incr}_n}{1}$ and apply the
$\beta$-reduction to reveal a modal substitution for the variable $u$
into $(\sappl{ \shandlesimp{u}{\textit{handlerSt}}{0}
}{x'}{\ret{x'}})$. The latter represents the scope of the initial
\texttt{let-box} term, expanding the syntactic sugar for statements.

In the next step~(\ref{eq:incr_reduction_4}), the modal substitution applies
handling with $\textit{handlerSt}$ with given state over the application of the
empty handling sequence on $u$. This proceeds by the following characteristic
definition case for modal substitution.
\[
  \substm{(\happl{u}{\Theta}{h}{e}{x'}{c'})}{\Psi}{c}{u} =
  \substc{
    \hndl{
      \hndlseq{
        c
      }{
        \substm{\Theta}{\Psi}{c}{u}
      }
    }{
      \substm{h}{\Psi}{c}{u}
    }{
      \substm{e}{\Psi}{c}{u}
    }
  }{x'}{
    (\substm{c'}{\Psi}{c}{u})
  }
\]
In other words, we first substitute $c$ for $u$ in all the
subterms. As customary, we ensure that $c$ doesn't contain free
occurrences of variable $x'$ (thus incurring capture), by
$\alpha$-renaming $x'$ if necessary.  In the expression
(\ref{eq:incr_reduction_1}), $u$ doesn't occur in the subterms, thus
the modal substitutions are vacuous, and we directly obtain
(\ref{eq:incr_reduction_4}).  


In the next step~(\ref{eq:incr_reduction_5}), the empty handling
sequencing immediately resolves according to the definition in
Figure~\ref{fig:hndl_seq_def}, and proceeds to handling with
\textit{handlerSt}.

Handling by \textit{handlerSt} in~(\ref{eq:incr_reduction_5}) results
in the computation $\ret{(0, 1)}$, as we shall illustrate shortly. In
the next step~(\ref{eq:incr_reduction_6}), this computation is
monadically substituted for $x'$ in $\ret{x'}$ to get the final
reduction result of $\ret{(0, 1)}$. This follows by the definition of
the monadic substitution in Figure~\ref{fig:substc_def}.
Monadic substitution $\substc{c'}{x}{c}$ sequentially precomposes $c'$
before $c$. In the special case when $c'$ is $\ret{e}$, it simply
substitutes $e$ for $x$ in $c$.

In step~(\ref{eq:incr_get_set_start}) we return to the
step~(\ref{eq:incr_reduction_5}) of applying the handler
\textit{handlerSt}. We proceed according to the definition of handling
from Figure~\ref{fig:hndl_def}. In particular, we obtain the
expression~(\ref{eq:incr_get_set_unfold}) by substituting the
following into the $\text{get}$ clause of \textit{handlerSt}:
\begin{itemize}
\item The supplied initial state $0$ is substituted for the state
  variable $z$.
\item The input argument $\tunit$ of \textit{get} is substituted for
  the variable $x$ in the handler clause for $\textit{get}$.
\item The \emph{handled} remainder of the initial computation is
  substituted for the continuation variable $k$, according to
  continuation substitution from Figure~\ref{fig:substk_def}.
\end{itemize}

The last component above, namely the handling of the remainder of the
computation is depicted in
lines~(\ref{eq:incr_set_start}-\ref{eq:incr_ret}) in
Figure~\ref{fig:incr_reduction}. We do not comment on it in any detail
as it is similar to the rest of the example, but we just note that it
obtains the result $\ret{(y, y+1)}$ in the scope of variables $y$ and
$z'$. Thus, it illustrates that all of our subsidiary operations work
over expressions with free variables.

We next execute the substitutions over $z$ and $x$ to derive
step~(\ref{eq:incr_get_k_red_1}). Next, to
obtain~(\ref{eq:incr_get_k_red_2}), we proceed with continuation
substitution whereby the arguments $0$ and $0$ of $k$ in the main
expression are replaced for $y$ and $z'$ respectively in the
expression being substituted for $k$. The remainder of the reduction
is then easy to complete.
 
\begin{figure}
  \input{formal/semantics_exp}
  \input{formal/semantics_comp}
  \Description{Operational semantics.}
  \caption{Call-by-value operational semantics. The semantics
    considers closed terms (no free variables) in the empty effect
    theory (no unhandled operations). Consequently, it only evaluates
    computations of the form $\ret{e}$ and $\letbox{u}{e}{c}$, as
    other cases involve free variables or operations. The rules
    include only those for $\beta$-reduction, and for enforcing
    call-by-value/left-to-right evaluation order.  The values of
    expression kind (ranged over by $v$) are $\lam{e}{A}{e'}$ and
    $\tbox{\Psi}{c}$. The value of computation kind is
    $\ret{v}$.}\label{fig:op_sem}
\end{figure}


\section{Soundness}\label{sec:soundness}

In this section we present the basic theoretical properties of \cmtte,
leading to the proofs of local soundness (i.e., that $\beta$-reduction
is well typed), local completeness (i.e., that $\eta$-expansion is
well typed), as well as the type soundness theorems (i.e., progress
and preservation) for the operational semantics from
Figure~\ref{fig:op_sem}.
The bulk of the development involves lemmas about the subsidiary
operations from Section~\ref{sec:reductions}. These lemmas all have
the form reminiscent of substitution principles, as they describe how
a variable or an effect operation can be replaced in a term.
Appendix~\ref{sec:theorems}
contains a more detailed presentation of
the proofs.

\subsection{Structural Properties}

To avoid writing out all $20$ combinations of our variable and term
judgements, we state weakening with generic binding forms.  We use
$x_\Delta: J_{\mathit{var}}$ to stand over variable bindings in
$\Delta$, namely $\tbind{x}{A}$ and $\mhypbind{u}{A}{\Psi}$.
$x_\Gamma : J_{\textit{var}}$ ranges over $\Gamma$ variable bindings
$\chypbind{\textit{op}}{A}{B}$ and $\khypbind{k}{A}{B}{C}$.  The
generic term judgement $\jdgmt{\Delta}{\Gamma}{t : J_{\mathit{term}}}$
ranges over the judgements for
$\tbind{e}{A}$ (in this case we implicitly ignore $\Gamma$), $\cbind{c}{A}$,
$\sbind{s}{A}$, $\hbind{h}{A}{\Psi}{S}{B}$, and $\hseqbind{\Theta}{A}{\Psi}{B}$.
We also use the generic term judgement to state the principles governing the
subsidiary operations from Section~\ref{sec:reductions}.

\begin{lemma}[Weakening]\label{lem:weakening}
  If $\jdgmt{\Delta}{\Gamma}{t : J_{\textit{term}}}$, then
  \begin{enumerate}
  \item $\jdgmt{\Delta, x_\Delta: J_{\mathit{var}}}{\Gamma}{t : J_{\textit{term}}}$.
  \item $\jdgmt{\Delta}{\Gamma, x_\Gamma:J_{\mathit{var}}}{t : J_{\textit{term}}}$.
  \end{enumerate}
\end{lemma}

\begin{lemma}[Expression substitution principle]\label{lem:subst_expr}
  If $\pjdgmt{\Delta}{e}{A}$
  and $\jdgmt{\Delta, \tbind{x}{A}}{\Gamma}{t : J_{\mathit{term}}}$,
  then $\jdgmt{\Delta}{\Gamma}{\subst{t}{e}{x} : J_{\mathit{term}}}$.
\end{lemma}

\begin{lemma}[Monadic substitution principle]\label{lem:substc_snd}
  If
  $\jdgmt{\Delta}{\Gamma}{\cbind{c}{A}}$
  and
  $\jdgmt{\Delta, \tbind{x}{A}}{\Gamma}{\cbind{c'}{B}}$,
  then
  $\jdgmt{\Delta}{\Gamma}{\cbind{\substc{c}{x}{c'}}{B}}$.
\end{lemma}

\begin{lemma}[Continuation substitution principle]\label{lem:subst_cont_handler_comp}
  If $\jdgmt{\Delta, \tbind{x}{A_1}, \tbind{y}{A_2}}{\Gamma}{\cbind{c'}{A_3}}$,
  then
  \begin{enumerate}
  \item
    If $\jdgmt{\Delta}{\Gamma, \khypbind{k}{A_1}{A_2}{A_3}}{\cbind{c}{B}}$,
    then $\jdgmt{\Delta}{\Gamma}{\cbind{\substkxy{c}{c'}{k}}{B}}$.
  \item
    If $\jdgmt{\Delta}{\Gamma, \khypbind{k}{A_1}{A_2}{A_3}}{
      \hbind{h}{C}{\Psi}{S}{C'}
    }$,
    then $\jdgmt{\Delta}{\Gamma}{
      \hbind{\substkxy{h}{c'}{k}}{C}{\Psi}{S}{C'}
    }$.
  \end{enumerate}
\end{lemma}

\begin{lemma}[Handling principle]\label{thm:hndl_snd}
  If $\jdgmt{\Delta}{\Psi}{\cbind{c}{A}}$,
  and $\jdgmt{\Delta}{\Gamma}{\hbind{h}{A}{\Psi}{B}{C}}$,
  and $\pjdgmt{\Delta}{e}{B}$,
  then $\jdgmt{\Delta}{\Gamma}{\cbind{\hndl{c}{h}{e}}{C}}$.
\end{lemma}

\begin{lemma}[Handling sequencing principle]\label{thm:hndl_seq_snd}
    If $\jdgmt{\Delta}{\Psi}{\cbind{c}{A}}$
    and $\jdgmt{\Delta}{\Psi'}{\hseqbind{\Theta}{A}{\Psi}{B}}$,
    then $\jdgmt{\Delta}{\Psi'}{\cbind{\hndlseq{c}{\Theta}}{B}}$.
\end{lemma}

\begin{lemma}[Modal substitution principle]\label{lem:substm_snd}
  If $\jdgmt{\Delta}{\Psi}{\cbind{c}{A}}$
  and $\jdgmt{\Delta, \mhypbind{u}{A}{\Psi}}{\Gamma}{t : J_{\mathit{term}}}$,
  then $\jdgmt{\Delta}{\Gamma}{\substm{t}{\Psi}{c}{u} : J_{\mathit{term}}}$.
\end{lemma}

\begin{lemma}[Identity handler]\label{lem:id_handler_snd}
  $\jdgmt{\Delta}{\Psi}{\hbind{\handid{\Psi}}{A}{\Psi}{\unitt{}}{A}}$.
\end{lemma}

\subsection{Main Theorems}

\begin{theorem}[Local soundness]\label{lem:beta_red}
If $\jdgmt{\Delta}{\Psi}{\cbind{c}{A}}$, then the following
$\beta$-reductions are well typed.
\begin{enumerate}
\item If $\jdgmt{\Delta}{\Gamma}{\cbind{\letboxu{\tbox{\Psi}{c}}{c'}}{B}}$,
  then $\jdgmt{\Delta}{\Gamma}{\cbind{\substm{c'}{\Psi}{c}{u}}{B}}$.
\item If $\pjdgmt{\Delta}{\letboxu{\tbox{\Psi}{c}}{e'}}{B}$,
then $\pjdgmt{\Delta}{\substm{e'}{\Psi}{c}{u}}{B}$.
\end{enumerate}
\end{theorem}

\begin{theorem}[Local completeness]\label{lem:eta_box_snd}
  If $\pjdgmt{\Delta}{e}{\boxt{\Psi}{A}}$,
  then the $\eta$-expansion of $e$ is well-typed, i.e.
$\pjdgmt{\Delta}{
      \letbox{u}{e}{
        \tbox{\Psi}{
          \shandlesimp{u}{\handid{\Psi}}{\tunit{}}
        }
      }
    }{\boxt{\Psi}{A}}$.
\end{theorem}

\begin{theorem}[Preservation on expressions]\label{thm:preservation_exp}
  If $\vdash {\tbind{e}{A}}$ and $e \reduces{} e'$, then
  $\vdash {\tbind{e'}{A}}$.
\end{theorem}

\begin{theorem}[Preservation on computations]
  If $\vdash {\cbind{c}{A}}$ and $c \reduces{} c'$, then
  $\vdash {\tbind{c'}{A}}$.
\end{theorem}

\begin{theorem}[Progress on expressions]\label{thm:progress_exp}
  If $\vdash {\tbind{e}{A}}$, then either
  (1) $e$ is a value, or
  (2) there exists $e'$ s.t. $e \reduces{} e'$.
\end{theorem}

\begin{theorem}[Progress on computations]
  If $\vdash {\cbind{c}{A}}$, then either
  (1) $c$ is $\ret{v}$ where $v$ is a value, or
  (2) there exists $c'$ s.t. $c \reduces{} c'$.
\end{theorem}

The above theorems directly imply \emph{effect safety}: 
programs with no effects (empty context $\Gamma$), can't incur
unhandled effect operation during execution. Indeed, preservation
implies that execution, which is always attempted on terms with no
effect operations, can't reach a computation that will use one.
In turn, progress implies that a computation with no
operations in $\Gamma$ can always make a step; in particular, it can't
get stuck on an unhandled operation.








\section{Extensions to \cmtte}\label{sec:extensions}

\begin{figure}
  \input{formal/extensions}
  \Description{Extension of \cmtte with \texttt{eval} and \texttt{let-fix}.}
  \caption{Extension of \cmtte with \texttt{eval} and \texttt{let-fix}.}\label{fig:extensions}
\end{figure}

\subsection{Evaluation}
In any calculus that tracks effects in types, one wants to relate the
category of computations with no effects to the purely functional
expressions. Similarly, in \cmtte we also would like to have a
correspondence between the types $A$ and $\Box A$.  In Section
\ref{sec:typing} we presented a function (monadic unit) of type
$\funt{A}{\Box A}$ that realizes one side of the correspondence.  To
establish the other, we need a function of type $\funt{\Box A}{A}$,
known in modal logic as axiom $T$ or reflexivity
\cite{blackburn:modal_logic}, as counit for the $\Box$ comonad in
categorical semantics \cite{bierman2000}, and as the \textit{eval}
function in modal type systems \cite{Davies01jacm}.

However, currently it's impossible to express in \cmtte, because modal
types can only be handled within the computation judgement.  That is,
we can always handle $\Box A$ into $A$, but only within the scope of a
box term.  Thus, we can write a function for $\funt{\Box A}{\Box A}$
\[
\tbind{
  \lam{x}{\Box A}{
    \tboxemp{
      \letboxu{x}{
        \shandlesimp{u}{\handid{\ectxsym}}{\tunit{}}
      }
    }
  }
}{\funt{\Box A}{\Box A}}
\]
but not for $\funt{\Box A}{A}$.  To support the latter, we need a way
to convert computations in an empty theory into expressions.  Thus we
present an extension of \cmtte with a primitive \texttt{eval}
expression.

We extend the language with a new expression $\teval{\Theta}{u}$ that,
similarly to handling, holds a handling sequence $\Theta$ applied to
$u$ (our extensions are summarized on Figure \ref{fig:extensions}).
For \texttt{eval}, the typing rule restricts $\Theta$ to typecheck in an empty
effect context which ensures that that the last clause of $\Theta$ will not use
any operation.
Thus the computation produced after applying this handling sequence will also be
operation-free.
This allows us to define the evaluation on empty effect context computations and
we additionally extend the modal substitution with the corresponding clause.

We establish the soundness of this addition by proving the following lemma:
\begin{lemma}[Eval principle]\label{lem:eval_snd}
  If $\jdgmt{\Delta}{\ectxsym}{\cbind{c}{A}}$, then
  $\pjdgmt{\Delta}{\heval{c}}{A}$.
\end{lemma}

Finally, this extension allows us to define \textit{eval} function in \cmtte and
so obtain the other side of the correspondence between $A$ and $\Box A$:
\[
\mathit{eval_f}
\eqdef
\tbind{
  \lam{x}{\Box A}{
    \letbox{u}{x}{
      \teval{\hseqbase}{u}
    }
  }
}{\funt{\Box A}{A}}
\]

The addition of \textit{eval} is relatively cheap and doesn't affect any of the
definitions from Figure \ref{fig:auxdef}; the modal substitution is the only
definition that we change.
Accordingly, the changes to the soundness proofs are also minimal.
The modal substitution principle for \textit{eval} readily follows from Lemmas
\ref{lem:eval_snd} and \ref{thm:hndl_seq_snd}.
And the progress and preservation theorems do not need to consider this case
explicitly because \texttt{eval} expression is not a closed term.

\subsection{Fixed Points}
Currently, there is no way to write loops in \cmtte.  In that it's
more similar to a system for logical inference than a programming
language.  We can however add loops through recursion to \cmtte.  We
introduce a fixed point operator \texttt{let-fix} and update
operational semantics accordingly (changes are summarized in Figure
\ref{fig:extensions}). The idea is that a recursive function $f$
should have the type $A \rightarrow \boxt{\Psi}B$, allowing it to use
effect operations from $\Psi$. As customary, the reduction in
operational semantics substitutes $f$ with its definition.
%

Thus \texttt{let-fix} allows us to write programs that use recursion, for example
a program that computes $3!$ and reduces to $6$ as we expect:
\begin{align*}
&\tfix{\textit{fact}}{n}{\intt{}}{\ectxsym}{
  \\
  &
  \quad
  \cifel{n = 0}{
    1
  }{
    n * (\appl{\mathit{eval_f}}{(\appl{\textit{fact}}{(n - 1)})})
  }
}{
  \\
  &
  \appl{\mathit{eval_f}}{(\appl{\textit{fact}}{3})}
}
\end{align*}

In
Appendix~\ref{sec:sub_ctx_var}
we present the definitions for all our subsidiary operations extended with
\texttt{eval} and \texttt{let-fix}, and we also include these cases in the
proofs in
Appendix~\ref{sec:theorems}. 

\section{Related and future work}

\subsubsection*{Type-and-effect Systems for Algebraic Effects}
Many languages for algebraic effects and handlers include effect
systems to track operations.  Among them are
$\lambda_\textnormal{eff}$ \cite{kammar2013}, which puts collection of
available operations $E$ in the typing judgement for computations
$\Gamma \vdash_E M : C$, similarly as we do with effect contexts, and
Eff \cite{bauer2014:effectsys,bauer2015,karachalias2020}, Links
\cite{hillerstrom2016,hillerstrom2018,hillerstrom2020}, Helium
\cite{biernacki2019}, Frank \cite{convent2020}, Koka
\cite{leijen2014}, and Effekt \cite{brachthauser2020oopsla} which put the
annotation on the computation's type.

In Eff, the type $A ! \Delta$ specifies that a computation might
invoke operations from $\Delta$, which is a collection of operations.
Links, Helium, and Koka rely on a more involved effect annotation
called effect rows, although their primary role---to track the use of
operations---is similar to the one in Eff.  Links for example
similarly specifies their computation type as $A ! E$ but there $E$ is
an effect row of a computation.  Notably, effect row languages support
a notion of row polymorphism, that allows a programmer to write
computations polymorphic over effect theories.

Effekt and the system of \citet{zhang2019}, also rely on a capability-passing
style for tracking effects.
In this style, effects are capabilities that handlers provide to the code within
their handling scope.
In contrast, the operations in \cmtte are \emph{variables}, which are bound by
\texttt{box}, and handled by the handler associated with the modal variable in
the corresponding \texttt{let-box}.
This distinction should be important for a future extension of \cmtte with
effect polymorphism, as we discuss below.

\cmtte further differs from all the above languages in the key property
that it manipulates effect annotations explicitly through introduction
and elimination forms of the modal type $\boxt{\Psi}{A}$.  Thus, the
\cmtte type-and-effect system is in fact just a type system, as the
effect annotations are built directly into the modal types. This leads
to an operational semantics determined solely by $\beta$-reduction,
although $\beta$-reduction itself now includes handling. More
importantly, we can consider the modal (effectful) types as
propositions, thus, in the future, obtaining a Curry-Howard style
interpretation for algebraic effects and handlers.

On the other hand, the above languages support more advanced features
in handlers than \cmtte, for example shallow \cite{hillerstrom2018}
and recursive \cite{bauer2015} handlers, or asynchronous effects, as
in $\lambda_{\text{\ae}}$ \cite{ahman2021}. We leave considering these
features for \cmtte for future work.

\subsubsection*{Scaling to Dependent Types}

\citet{ahman2018} develops a dependently typed calculus for algebraic
effects based on call-by-push value approach, called \textsc{e}MLTT.
In this paper we don't consider dependent types. However, the
foundation of \cmtte in modal logic should offer some conceptual
simplification when scaling to dependent types, as we can build on the
recent work on modal type theories.
For example, the modal foundation led us to define handling as part of
$\beta$-reduction on open terms. In \textsc{e}MLTT, handling also
works over open terms, but is governed by several definitional
equalities which describe how handling reduces.  In contrast, \cmtte
requires only a single $\beta$-reduction of \texttt{let-box} which
encompasses handling.

CMTT itself has been lifted to dependent types~\cite{nanevski2008},
albeit in a somewhat restricted form, with abstraction over modal
variables, but no explicit modality.  More
recently,~\citet{gratzer:icfp19} have developed a variant of
Martin-L\"of type theory with a non-contextual $\Box$ modality, while
using an alternative elimination form
$\texttt{unbox}$~\cite{Davies01jacm} instead of $\texttt{let-box}$ to
avoid commuting conversions. We might consider a similar approach in future
work.

Currently, \cmtte admits only free algebraic theories; i.e. those with
no equations between operations. \citet{luksic2020} add support for
equations in Eff through the extension of the annotations to $A !
\Delta / \mathcal{E}$, where $\mathcal{E}$ stands for the set of
equations.  \citet{ahman2018} supports equations and also treats them
as a special annotation over the computation types.  We expect that
the scaling to dependent types will immediately allow us to support
equations, or any arbitrary propositions over operations. Indeed, in
dependent type theories, propositions are merely types, and thus can
easily be added to a context $\Psi$ representing an algebraic theory
in a modal type.

\subsubsection*{Effect Polymorphism and Abstraction Safety}
The work related to
Beluga~\cite{pientka:flops10,pientka:popl08,cave:lfmtp13} has extended
CMTT, among other features, with abstraction over first-class contexts
and explicit substitutions. In the future, we plan to incorporate
similar extensions to \cmtte and apply them to algebraic
effects. These should provide new and effective solutions to the
problem of unintended capture of operations in the presence of effect
polymorphism.

To describe the problem, suppose we have a handler that handles only a
single operation \textit{op}.  We want to apply this handler to an
effect polymorphic computation that has an operation signature $[E,
  \textit{op}]$, where $E$ is a variable that stands for some set of
effects. After handling, $E$ is supposed to remain in the signature of
the resulting computation, while \textit{op} is to be stripped away by
handling. The problem of unintended capture is to ensure that $E$
remains in the signature after handling, even if we instantiate $E$
with \textit{op} itself.

Current solutions in the field of algebraic effects rely on a coupling
of operations with their intended handlers through lexical scoping
\cite{zhang2019,biernacki2020,brachthauser2020oopsla}. They are
implemented by making handlers emit a capability, or a label, that is
then passed to the code in the handler's scope and is attached to each
operation call so that handlers can dynamically distinguish the
operations they know of from others.

In the field of contextual modal type theory, the problem of
unintended capture of variables in the presence of context
polymorphism was already considered by \citet{pientka:popl08} in
Beluga.  Suppose we extended \cmtte following this approach, and let
us sketch how that problem would look in the extended system where we
can bind contexts, and what solution we plan to borrow:
\begin{align*}
&
\lam{\psi}{\textit{ctx}}{
  \lam{f}{\funt{\unitt{}}{
      \boxt{\psi, \chypbind{\textit{op}}{\intt{}}{\intt{}}}{\intt{}}
  }}{
    \tbox{\psi}{
      \\
      &
      \quad
      \letboxu{
        \appl{f}{\tunit{}}
      }{
        \\
        &
        \quad
        \shandlesimp{u}{
          (\handid{\psi},
          \handoparr{\textit{op}}{x}{k}{z}{
            \sscont{k}{(x + 1)}{z}
          })
        }{\tunit{}}
      }
    }
  }
}
\end{align*}
The example program above binds some concrete context to $\psi$, takes
a function $f$ that produces a computation in the theory $(\psi,
\textit{op})$, then runs $f$, and handles only \textit{op},
propagating any eventual operations from $\psi$ through the identity
handler $\handid{\psi}$.
The problem of unintended capture appears if we consider instantiating
$\psi$ with some context that itself contains
$\chypbind{\textit{op}}{\intt{}}{\intt{}}$, as this operation will now
be handled rather than propagated, contrary to the intended semantics.

One can see that the example is problematic already from the types, as
the substitution of $\psi$ causes the problem in the context of $f$,
which now contains two instances of \textit{op}, one shadowing the
other. To avoid the shadowing, and thus also the unintended capture,
Beluga locally $\alpha$-renames the variables in the context
instantiating $\psi$. As this context stands for the (unknown) bound
variables in the type and body of $f$, $\alpha$-renaming them doesn't
change the semantics. The renaming occurs in the type of $f$ but also
in the index of the identity handler $\handid{\psi}$. Thus, once
$\alpha$-renamed, the occurrence of \textit{op} arising out of $\psi$
is given to the identity handler for propagation, as intended, rather
than being captured by the \textit{op} clause in the handler for $u$.
%
We will follow this approach in extending \cmtte with effect
polymorphism.
%
%
We will also consider how one efficiently finds the nearest handler in
a handling sequence that treats a given effect in a non-trivial way
(i.e., not by simply propagating it via identity
handler)~\cite{schuster2020,xi:icfp2020}.

\subsubsection*{Relationship to Comonads and Modal Logic and Calculi}

As illustrated in Section~\ref{sec:reductions}, our formulation of
\cmtte employs the $\letbox{u}{e_1}{c_2}$ constructor. When $e_1$
itself is a value---thus, by typing, necessarily of the form
$\tbox{\Psi}{c_1}$---the reduction modally substitutes $c_1$ for $u$
in $c_2$. However, whether $c_1$ evaluates depends on the occurrences
of $u$ in $c_2$. This is in contrast to the elimination rule for
monads (i.e., the monadic bind), where the bound computation
immediately executes.

This kind of reduction is characteristic of comonadic calculi. For
example, CMTT was given comonadic semantics
in~\cite{gab+nan:japl13}. More generally, comonadic calculi for
coeffects~\cite{gab:icfp16,orch:icfp19} also employ constructs with
similar behavior, inspired by the elimination rule for
exponentials in linear $\lambda$-calculus. However, to the best of our
knowledge, these calculi haven't been applied to algebraic effects
and handling. On a related note, CMTT has been applied to staged
computation and meta programming~\cite{Davies01jacm,Nanevski05jfp},
and recently it was proposed that staging could be useful for modular treatment
of algebraic effects~\cite{staged:pepm21,wei2020,schuster2020}, albeit also
without using modal logic.

\citet{nanevski:phd04} presents modal calculi similar to CMTT, where
the graded $\Box$ modality tracks effects that depend on the execution
environment but don't change it. These are handleable effects;
examples include exceptions and delimited
continuations~\cite{nanevski03cmu}, and dynamic
binding~\cite{nanevski03ppdp}. In contrast, \cmtte supports algebraic
effects whose execution may change the environment (e.g., witness the
theory $\textit{St}$ of state in Section~\ref{sec:overview}), and
requires a significantly more general notion of handling. Another
distinction is that \cmtte models effects simply as contexts of
operations, whereas loc.~cit.~use much more involved freshness and
binding disciplines inspired by Nominal logic~\cite{pitts2003}.

Moreover, as we show in the examples (11-13) in
Section~\ref{sec:typing}, the modality $\boxt{\Psi}{A}$ exhibits the
functions with the types that are in programming usually associated
with monads. Furthermore, the typing $\Box A \rightarrow A$ of the
comonadic counit becomes available with the addition of the evaluation
construct in Section~\ref{sec:extensions}.
In the future we will study the equational theory of \cmtte, with the
goal of clarifying the exact categorical nature of our modality. A
useful step in this direction will be the work on categorical
semantics of operations in scope~\cite{pirog:lics18}.

Another fruitful research direction is suggested by~\citet{wu2014} in
Haskell. It illustrates that the practical use of algebraic effects
requires operations with types which are higher-order, meaning that
the operation can be parametrized by a computation, or even a
handler. The former is already possible in \cmtte, as our operations
can range over modal types.  The latter will require internalization
of effect handlers.  While we have a judgment for handlers, we don't
currently have a type for them that can be combined with other types.
In CMTT this corresponds to internalizing the judgment for explicit
substitutions, which has been done in Beluga by~\citet{cave:lfmtp13}.

We also plan to study how graded $\Diamond$ modality (called
''possibility'', or ''diamond'') can be used in
\cmtte. \citet{nanevski03ppdp} proposed that the proof terms
for $\langle\Psi\rangle$ correspond to installing a default handler
for a number of effect names. Afterwards, these effects do not need to
be explicitly handled, as the default handler applies. We expect that
similar behavior will usefully extend to algebraic effects in \cmtte.

Finally, the type system of \cmtte makes sense as a logic as well,
i.e., when one erases the terms and just considers the types as
propositions. In the future we plan to study this logic (e.g., its Kripke
semantics, its sequent calculus, etc.) and derive correspondence with \cmtte in
the style of Curry and Howard.

\section{Conclusion}\label{sec:conclusion}

We have presented the design of \cmtte, a novel contextual modal
calculus for algebraic effects and handlers. We start from the idea
that an algebraic theory can be represented as a variable context, and
apply the graded modal necessity type $\boxt{\Psi}A$ to classify
computations of type $A$ that may invoke effects described by the
algebraic theory (equivalently, context) $\Psi$. The notion of
handling naturally arises as a way to transform $\Psi$ into another
theory. To the best of our knowledge, this is the first calculus that
relates algebraic effects to modal types.

\cmtte is organized around
$\beta$-reduction and $\eta$-expansion for its type
constructors. Defining these requires developing interesting technical
concepts such as identity handlers and modal substitutions that
encompass handling. We illustrated the system on a number of examples,
and established the basic soundness properties.

\begin{acks}
We are grateful to Mi\"etek Bak, Alexander Gryzlov, and Patrick Cousot
for their comments and discussions on the earlier presentation of the
work.  We thank the anonymous reviewers from ICFP'21 PC and our
shepherd Brigitte Pientka for their feedback, which significantly
helped us to understand and describe our contribution better.

This research was partially supported by the Spanish MICINN projects
BOSCO (PGC2018-102210-B-I00) and
the European Research Council project Mathador
(ERC2016-COG-724464). Any opinions, findings, and conclusions or
recommendations expressed in the material are those of the authors and
do not necessarily reflect the views of the funding agencies.
\end{acks}


\bibliography{paper}

\appendix
\section{Full definitions of subsidiary operations (including eval and let-fix)}\label{sec:sub_ctx_var}
\subsection{Modal Substitution}
\begin{align*}
\substm{x}{\Psi}{c}{u} &=~ x
\\
\substm{(\lam{x}{A}{e})}{\Psi}{c}{u} &=~ \lam{x}{A}{\substm{e}{\Psi}{c}{u}}
\\
\substm{(\appl{t_1}{t_2})}{\Psi}{e}{u} &=~
  \appl{\substm{t_1}{\Psi}{e}{u}}{\substm{t_1}{\Psi}{e}{u}}
\\
\substm{(\tbox{\Gamma}{c'})}{\Psi}{c}{u} &=~
  \tbox{\Gamma}{\substm{c'}{\Psi}{c}{u}}
\\
\substm{(\letbox{u'}{e_1}{e_2})}{\Psi}{c}{u} &=~
  \letbox{u'}{\substm{e_1}{\Psi}{c}{u}}{\substm{e_2}{\Psi}{c}{u}}
    \quad \text{s.t. $u \neq u'$}
\\
\substm{(\teval{\Theta}{u})}{\Psi}{c}{u} &=~
  \heval{{\hndlseq{c}{\substm{\Theta}{\Psi}{c}{u}}}}
\\
\substm{(\teval{\Theta}{u'})}{\Psi}{c}{u} &=~
  \teval{\substm{\Theta}{\Psi}{c}{u}}{u'}
    \quad \text{s.t. $u \neq u'$}
\\
\substm{(\tfix{f}{x}{A}{\Psi'}{c'}{e})}{\Psi}{c}{u} &=~ 
\\
  \tfix{&f}{x}{A}{\Psi'}{\substm{c'}{\Psi}{c}{u}}{\substm{e}{\Psi}{c}{u}}
\\
&
\\
\substm{(\ret{e})}{\Psi}{c}{u} &=~ \ret{\substm{e}{\Psi}{c}{u}}
\\
\substm{(\compappl{\textit{op}}{e}{x}{c'})}{\Psi}{c}{u} &=~
  \compappl{\textit{op}}{\substm{e}{\Psi}{c}{u}}{x}{\substm{c'}{\Psi}{c}{u}}
\\
\substm{(\contappl{k}{e_1}{e_2}{x}{c'})}{\Psi}{c}{u} &=~
  \contappl{k}{
    \substm{e_1}{\Psi}{c}{u}
  }{
    \substm{e_2}{\Psi}{c}{u}
  }{x}{
    \substm{c'}{\Psi}{c}{u}
  }
\\
\substm{(\happl{u}{\Theta}{h}{e}{x'}{c'})}{\Psi}{c}{u} &=~
  \substc{
    \hndl{
      \hndlseq{
        c
      }{
        \substm{\Theta}{\Psi}{c}{u}
      }
    }{
      \substm{h}{\Psi}{c}{u}
    }{
      \substm{e}{\Psi}{c}{u}
    }
  }{x'}{
    \substm{c'}{\Psi}{c}{u}
  }
\\
\substm{(\happl{u'}{\Theta}{h}{e}{x'}{c'})}{\Psi}{c}{u} &=~
  \\
  \happl{u'}{
    \substm{\Theta&}{\Psi}{c}{u}
  }{
    \substm{h}{\Psi}{c}{u}
  }{
    \substm{e}{\Psi}{c}{u}
  }{x'}{
    \substm{c'}{\Psi}{c}{u}
  }
  \quad \text{where $u \neq u'$}
\\
\substm{(\letbox{u'}{e}{c'})}{\Psi}{c}{u} &=~
  \letbox{u'}{\substm{e}{\Psi}{c}{u}}{\substm{c'}{\Psi}{c}{u}}
    \quad \text{where $u \neq u'$} 
\\
\substm{(\tfix{f}{x}{A}{\Psi}{c_1}{c_2})}{\Psi}{c}{u} &=~ 
  \\
  \tfix{&f}{x}{A}{\Psi}{\substm{c_1}{\Psi}{c}{u}}{\substm{c_2}{\Psi}{c}{u}}
\\
&
\\
\substm{(\hseqbase{})}{\Psi}{c}{u} &=~ \hseqbase{}
\\
\substm{(\hseqextend{\Theta}{h}{e}{x'}{c'})}{\Psi}{c}{u} &=~
  \hseqextend{
    \substm{\Theta}{\Psi}{c}{u}
  }{
    \substm{h}{\Psi}{c}{u}
  }{
    \substm{e}{\Psi}{c}{u}
  }{
    x'
  }{
    \substm{c'}{\Psi}{c}{u}
  }
\\
&
\\
\substm{(\handbase{x}{z}{c'})}{\Psi}{c}{u} &=~
  \handbase{x}{z}{\substm{c'}{\Psi}{c}{u}}
\\
\substm{(\hextend{h}{\handoparr{p}{y}{k}{z}{c'}})}{\Psi}{c}{u} &=~
  \hextend{\substm{h}{\Psi}{c}{u}}{
    \handoparr{p}{y}{k}{z}{\substm{c'}{\Psi}{c}{u}}
}
\end{align*}

\subsection{Monadic Substitution}
\begin{align*}
  \substc{\ret{e}}{x}{c} &=~ \subst{c}{e}{x} 
  \\
  \substc{\sappl{s}{y}{c'}}{x}{c} &=~ \sappl{s}{y}{\substc{c'}{x}{c}}
  \\
  \substc{\letbox{u}{e}{c'}}{x}{c} &=~ \letbox{u}{e}{\substc{c'}{x}{c}}
  \\
  \substc{\tfix{f}{x}{A}{\Psi}{c_1}{c_2}}{y}{c} &=
    \tfix{f}{x}{A}{\Psi}{c_1}{\substc{c_2}{y}{c}}
\end{align*}

\subsection{Continuation Substitution}
\begin{align*}
\substkxy{(\ret{e})}{c}{k} &=~ \ret{e} 
\\
\substkxy{(\compappl{op}{e}{x'}{c'})}{c}{k} &=~
  \compappl{op}{e}{x'}{\substkxy{c'}{c}{k}}
\\
\substkxy{(\contappl{k}{e_1}{e_2}{x'}{c'})}{c}{k} &=~
\substc{\subst{\subst{c}{e_1}{x}}{e_2}{y}}{x'}{c'}
\\
\substkxy{(\contappl{k'}{e_1}{e_2}{x'}{c'})}{c}{k} &=~
 \contappl{k'}{e_1}{e_2}{x'}{\substkxy{c'}{c}{k}}
 ~
 \text{when $k \neq k'$}
\\
\substkxy{(\letboxu{e}{c'})}{c}{k} &=~ \letboxu{e}{\substkxy{c'}{c}{k}}
\\
\substkxy{(\tfix{f}{x'}{A}{\Psi}{c_1}{c_2})}{c}{k} &=~
  \\
  \tfix{f}{x'}{&A}{\Psi}{c_1}{\substkxy{c_2}{c}{k}}
\\
\substkxy{(\happl{u}{\Theta}{h}{e}{x'}{c'})}{c}{k} &=~
  \\
  \happl{u}{&\Theta}{\substkxy{h}{c}{k}}{e}{x'}{\substkxy{c'}{c}{k}}
\\
&
\\
\substkxy{(\handbase{x'}{z}{c'})}{c}{k} &=~
  \handbase{x'}{z}{\substkxy{c'}{c}{k}}
\\
\substkxy{(\hextend{h}{\handoparr{op}{x'}{k'}{z}{c'}})}{c}{k} &=~
  \hextend{\substkxy{h}{c}{k}}{
    \handoparr{op}{x'}{k'}{z}{\substkxy{c'}{c}{k}}
  }
\end{align*}

\subsection{Handling}
\begin{align*}
  \hndl{\ret{e'}}{h}{e} &=~ \subst{\subst{c}{e'}{x}}{e}{z}
  \quad\text{where}~\handbase{x}{z}{c} \in h
  \\
  \hndl{\compappl{\textit{op}}{e'}{y}{c'}}{h}{e} &=~
    \\
    \substk{\subst{\subst{c_{\textit{op}}}{e}{z}}{&e'}{x'}}{y}{z'}{
      \hndl{c'}{h}{z'}
    }{k}
    \quad
    \text{where }
    (\handoparr{\textit{op}}{x'}{k}{z}{c_{\textit{op}}}) \in h
  \\
  \hndl{\tfix{f}{x}{A}{\Psi}{c_1}{c_2}}{h}{e} &=~
    \tfix{f}{x}{A}{\Psi}{c_1}{\hndl{c_2}{h}{e}}
  \\
  \hndl{\letboxu{e'}{c'}}{h}{e} &=~
  \letboxu{e'}{\hndl{c'}{h}{e}} 
  \\
  \hndl{\happl{u}{\Theta}{h'}{e'}{x'}{c'}}{h}{e} &=~
  \happl{u}{
    \hseqextend{\Theta}{h'}{e'}{x'}{c'}
  }{h}{e}{x}{\ret{x}} 
\end{align*}

\subsection{Handling Sequencing}
\begin{align*}
  \hndlseq{c'}{\hseqbase{}} &=~ c' 
  \\
  \hndlseq{c'}{\hseqextend{\Theta}{h}{e}{x}{c}} &=~
  \substc{\hndl{\hndlseq{c'}{\Theta}}{h}{e}}{x}{c}
\end{align*}

\subsection{Eval}
\begin{align*}
  \heval{\ret{e}} &= e
  \\
  \heval{\sappl{\shandle{u}{\Theta}{h}{e}}{x}{c}} &=
  \teval{\Theta, \hseqclause{h}{e}{x}{c}}{u}
  \\
  \heval{\letbox{u}{e}{c}} &= \letbox{u}{e}{\heval{c}}
  \\
  \heval{\tfix{f}{x}{A}{\Psi}{c_1}{c_2}} &=
    \tfix{f}{x}{A}{\Psi}{c_1}{\heval{c_2}}
\end{align*}

\section{Theorems and proofs}\label{sec:theorems}

\subsection{Structural Properties}

To avoid writing out all $20$ combinations of our variable and term
judgements, we state weakening with generic binding forms.  We use
$x_\Delta: J_{\mathit{var}}$ to stand over variable bindings in
$\Delta$, namely $\tbind{x}{A}$ and $\mhypbind{u}{A}{\Psi}$.
$x_\Gamma : J_{\textit{var}}$ ranges over $\Gamma$ variable bindings
$\chypbind{\textit{op}}{A}{B}$ and $\khypbind{k}{A}{B}{C}$.  The
generic term judgement $\jdgmt{\Delta}{\Gamma}{t : J_{\mathit{term}}}$
ranges over the judgements for
$\tbind{e}{A}$ (in this case we implicitly ignore $\Gamma$), $\cbind{c}{A}$,
$\sbind{s}{A}$, $\hbind{h}{A}{\Psi}{S}{B}$, and $\hseqbind{\Theta}{A}{\Psi}{B}$.
We also use the generic term judgement to state the principles governing the
subsidiary operations from Section~\ref{sec:reductions}.

\begin{lemma}[Weakening]\label{alem:weakening}
  If $\jdgmt{\Delta}{\Gamma}{t : J_{\textit{term}}}$, then
  \begin{enumerate}
  \item $\jdgmt{\Delta, x_\Delta: J_{\mathit{var}}}{\Gamma}{t : J_{\textit{term}}}$.
  \item $\jdgmt{\Delta}{\Gamma, x_\Gamma:J_{\mathit{var}}}{t : J_{\textit{term}}}$.
  \end{enumerate}
\end{lemma}
\begin{proof}
  The cases follow by mutual induction over the premise typing judgements and
  our assumption on the distinctness of named variables.
\end{proof}

\begin{lemma}[Expression substitution principle]\label{alem:subst_expr}
  If $\pjdgmt{\Delta}{e}{A}$
  and $\jdgmt{\Delta, \tbind{x}{A}}{\Gamma}{t : J_{\mathit{term}}}$,
  then $\jdgmt{\Delta}{\Gamma}{\subst{t}{e}{x} : J_{\mathit{term}}}$.
\end{lemma}
\begin{proof}
  By mutual induction on the expansions of the second premise.
  We present the interesting cases for expressions:
  \begin{description}
  \item[Case $e' = x$.] 
    The goal is to show
    $\pjdgmt{\Delta}{e}{B}$,
    while we assume
    $\pjdgmt{\Delta, \tbind{x}{A}}{x}{B}$
    and we have $\tbind{x}{B} \in \Delta, \tbind{x}{A}$ by inversion.
    Thus we know $B = A$ and the goal follows by assumption.

  \item[Case $e' = y$.] 
    We need to show
    $\pjdgmt{\Delta}{y}{B}$,
    while we know $\tbind{y}{B} \in \Delta, \tbind{x}{A}$ by inversion.
    The goal follows by the \textsc{var} rule, where showing $\tbind{y}{B} \in
    \Delta$ is straightforward.
     
  \item[Case $e' = \lam{y}{B_1}{e''}$.]
    We need to show
    $\pjdgmt{\Delta}{\lam{y}{B_1}{\subst{e''}{e}{x}}}{\funt{B_1}{B_2}}$
    for some $B_1$ and $B_2$, with $x \neq y$ given.
    By inversion on $e'$, we know that
    $\pjdgmt{\Delta, \tbind{x}{A}, \tbind{y}{B_1}}{e''}{B_2}$
    By the \textsc{${\rightarrow}I$} rule we need to show
    $\pjdgmt{\Delta, \tbind{y}{B_1}}{\subst{e''}{e}{x}}{B_2}$
    which follows by induction.
  \end{description}
\end{proof}

\begin{lemma}[Monadic substitution principle]\label{alem:substc_snd}
  If
  $\jdgmt{\Delta}{\Gamma}{\cbind{c}{A}}$
  and
  $\jdgmt{\Delta, \tbind{x}{A}}{\Gamma}{\cbind{c'}{B}}$,
  then
  $\jdgmt{\Delta}{\Gamma}{\cbind{\substc{c}{x}{c'}}{B}}$.
\end{lemma}
\begin{proof}
  By induction on $c$.
  Cases are:
  \begin{description}
  \item[Case $c = \ret{e}$.]
    The goal $\jdgmt{\Delta}{\Gamma}{\cbind{\subst{c'}{e}{x}}{B}}$ follows by
    the substitution lemma \ref{alem:subst_expr}.
  \item[Case $c = \sappl{s}{y}{c''}$.]
    The goal
    $\jdgmt{\Delta}{\Gamma}{\cbind{\sappl{s}{y}{\substc{c''}{x}{c'}}}{B}}$
    follows by the rule \textsc{bind}, weakening, and induction.
  \item[Case $c = \letbox{u}{e}{c''}$.]
    The goal
    $\jdgmt{\Delta}{\Gamma}{\cbind{\letbox{u}{e}{\substc{c''}{x}{c'}}}{B}}$
    follows by the rule \textsc{$\Box E$}, weakening, and induction.
  \item[Case $c = \tfix{f}{x}{A'}{\Psi}{c_1}{c_2}$]
    follows analogously to the case above.
  \end{description}
\end{proof}

\begin{lemma}[Continuation substitution principle]\label{alem:subst_cont_handler_comp}
  If $\jdgmt{\Delta, \tbind{x}{A_1}, \tbind{y}{A_2}}{\Gamma}{\cbind{c'}{A_3}}$,
  then
  \begin{enumerate}
  \item
    If $\jdgmt{\Delta}{\Gamma, \khypbind{k}{A_1}{A_2}{A_3}}{\cbind{c}{B}}$,
    then $\jdgmt{\Delta}{\Gamma}{\cbind{\substkxy{c}{c'}{k}}{B}}$.
  \item
    If $\jdgmt{\Delta}{\Gamma, \khypbind{k}{A_1}{A_2}{A_3}}{
      \hbind{h}{C}{\Psi}{S}{C'}
    }$,
    then $\jdgmt{\Delta}{\Gamma}{
      \hbind{\substkxy{h}{c'}{k}}{C}{\Psi}{S}{C'}
    }$.
  \end{enumerate}
\end{lemma}
\begin{proof}
  By mutual induction on the first premise of each branch.
  The interesting case is the continuation statement, where the continuation is
  the one being substituted:
  \begin{description}
    \item[Case $c = \contappl{k}{e_1}{e_2}{x'}{c''}$.]
      The goal is to show $\jdgmt{\Delta}{\Gamma}{\cbind{
          \substc{
            \subst{
              \subst{c'}{e_1}{x}
            }{e_2}{y}
          }{x'}{c''}
        }{B}}$,
      while by inversion on $c$ we assume $\jdgmt{\Delta,
        \tbind{x'}{A_3}}{\Gamma}{\cbind{c''}{B}}$,
      $\pjdgmt{\Delta}{e_1}{A_1}$, and
      $\pjdgmt{\Delta}{e_2}{A_2}$.
      We first derive $\jdgmt{\Delta}{\Gamma}{\cbind{
          \subst{
            \subst{c'}{e_1}{x}
          }{e_2}{y}
        }{A_3}}$ by the substitution principle (Lemma~\ref{alem:subst_expr}).
      Then the goal follows by the monadic substitution principle (Lemma~\ref{alem:substc_snd}).
  \end{description}
\end{proof}

\begin{lemma}[Handling principle]\label{alem:hndl_snd}
  If $\jdgmt{\Delta}{\Psi}{\cbind{c}{A}}$,
  and $\jdgmt{\Delta}{\Gamma}{\hbind{h}{A}{\Psi}{B}{C}}$,
  and $\pjdgmt{\Delta}{e}{B}$,
  then $\jdgmt{\Delta}{\Gamma}{\cbind{\hndl{c}{h}{e}}{C}}$.
\end{lemma}
\begin{proof}
By induction on $c$. The cases are:
\begin{description}
\item[Case $c = \ret{e'}$.] 
  By the definition of handling we need to show
  $\jdgmt{\Delta}{\Gamma}{\cbind{
      \subst{
        \subst{c'}{e'}{x}
      }{e}{z}
    }{C}
  }$
  s.t. $\handbase{x}{z}{c'} \in h$ and $\jdgmt{\Delta, \tbind{x}{A},
    \tbind{z}{B}}{\Gamma}{\cbind{c'}{C}}$.
  By inversion on typing of $c$ we know that
  $\pjdgmt{\Delta}{e'}{A}$.
  Thus we show the goal with the assumption
  $\pjdgmt{\Delta}{e}{B}$ and the substitution principle (Lemma
  \ref{alem:subst_expr}).
\item[Case $c = \compappl{\textit{op}}{e'}{x'}{c'}$.]
  By the definition of handling, we need to show:
  \[
  \jdgmt{\Delta}{\Gamma}{
    \cbind{
      \substk{
        \subst{
          \subst{
            c_{\textit{op}}
          }{e}{z}
        }{e'}{x}
      }{x'}{z'}{
        \hndl{c'}{h}{z'}
      }{k}
    }{C}
  }
  \]
  where $(\handoparr{op}{x}{k}{z}{c_{\textit{op}}}) \in h$ since $\textit{op}$
  is in $\Psi$ by typing.
  Thus by inversion on the typing of $c$, we have
  $\chypbind{\textit{op}}{A_1}{A_2} \in \Psi$ for some $A_1$ and $A_2$,
  $\pjdgmt{\Delta}{e'}{A_1}$, and $\jdgmt{\Delta,
    \tbind{x'}{A_2}}{\Psi}{\cbind{c'}{A}}$.
  Also, by the typing of the handler operation clauses, we know:
  $\jdgmt{
    \Delta, \tbind{x}{A_1}, \tbind{z}{B}
  }{
    \Gamma, \khypbind{k}{A_2}{B}{C}
  }{\cbind{c_{\textit{op}}}{C}}$.

  For substitution of $z$ and $x$ we apply Lemma~\ref{alem:subst_expr}, and
  for $k$ we use Lemma~\ref{alem:subst_cont_handler_comp}.
  The expression variable substitutions immediately follow, but the continuation
  substitution changes the goal to
  $\jdgmt{\Delta, \tbind{x'}{A_2}, \tbind{z'}{B}}{\Gamma}{\cbind{
      \hndl{c'}{h}{z'}
    }{C}}$
  which we show by induction on $c'$ and weakening.
\item[Case $c = \letbox{u}{e'}{c_n}$.]
  The goal is to show
  $\jdgmt{\Delta}{\Gamma}{\cbind{
      \letbox{u}{e'}{\hndl{c_n}{h}{e}}
    }{C}}$
  and it immediately follows by the \textsc{$\Box E$} rule, and induction.
\item[Case $c = \tfix{f}{x'}{A'}{\Psi'}{c_1}{c_2}$]
  follows analogously to the case above.
\item[Case $c = \happl{u}{\Theta}{h'}{e'}{x'}{c'}$.]
  The goal is to prove
  \[
  \jdgmt{\Delta}{\Gamma}{
    \cbind{
      \happl{u}{
        \Theta, \hseqclause{h'}{e'}{x'}{c'}
      }{h}{e}{x}{\ret{x}}
    }{C}
  }
  \]
  and by the inversion on the typing of $c$ we know that
  $\mhypbind{u}{A_1}{\Psi_1} \in \Delta$,
  $\jdgmt{\Delta}{\Psi_2}{\hseqbind{\Theta}{A_1}{\Psi_1}{A_2}}$,
  $\jdgmt{\Delta}{\Psi}{\hbind{h'}{A_2}{\Psi_2}{S_{h'}}}{A_3}$,
  $\pjdgmt{\Delta}{e}{S_{h'}}$, and
  $\jdgmt{\Delta, \tbind{x'}{A_3}}{\Psi}{\cbind{c'}{A}}$ for some types $A_1$,
  $A_2$, $A_3$ and algebraic theories $\Psi_1$, $\Psi_2$.

  Thus we use the \textsc{seqh} rule to add $h'$ into the sequence. We obtain:
  $\jdgmt{\Delta}{\Psi}{\hseqbind{
      \Theta,
      \hseqclause{h'}{e'}{x'}{c'}
    }
  }{A_1}{\Psi_1}{A}$.
  Then with the \textsc{mvar} we get
  $\jdgmt{\Delta}{\Gamma}{
    \cbind{
      \shandle{u}{
        \Theta,
        \hseqclause{h'}{e'}{x'}{c'}
      }{h}{e}
    }{C}
  }$ by our typing assumptions on $h$ and $e$.
  Then the case follows by \textsc{bind} and \textsc{ret}.
\end{description}
\end{proof}

\begin{lemma}[Handling sequencing principle]\label{alem:hndl_seq_snd}
    If $\jdgmt{\Delta}{\Psi}{\cbind{c}{A}}$
    and $\jdgmt{\Delta}{\Psi'}{\hseqbind{\Theta}{A}{\Psi}{B}}$,
    then $\jdgmt{\Delta}{\Psi'}{\cbind{\hndlseq{c}{\Theta}}{B}}$.
\end{lemma}
\begin{proof}
  By induction on $\jdgmt{\Delta}{\Psi'}{\hseqbind{\Theta}{A}{\Psi}{B}}$.
  The cases are:
  \begin{description}
  \item[Case $\Theta = \hseqbase$.]
    The goal $\jdgmt{\Delta}{\Psi'}{\cbind{c}{B}}$ follows from the typing of
    $\Theta$ and our assumptions.
    By the rule \textsc{emph} the input and output types of $\Theta$ must much,
    and thus $B = A$.
    From it, we also know that $\Psi \subseteq \Psi'$ and thus we apply
    weakening to finish the goal.
  \item[Case $\Theta = \Theta', \hseqclause{h}{e}{x'}{c'}$.]  By the
    typing of $\Theta$, we know:
    $\jdgmt{\Delta}{\Psi_h}{\hseqbind{\Theta'}{A}{\Psi}{A'}}$,
    $\jdgmt{\Delta}{\Psi'}{\hbind{h}{A'}{\Psi_h}{S}{A''}}$,
    $\pjdgmt{\Delta}{e}{S}$, and $\jdgmt{\Delta,
      \tbind{x'}{A''}}{\Psi'}{\cbind{c'}{B}}$.  The goal is to show
    $\jdgmt{\Delta}{\Psi'}{ \cbind{ \substc{ \hndl{
            \hndlseq{c}{\Theta'} }{h}{e} }{x'}{ c' } }{B} }$.  First,
    we apply the monadic substitution principle (Lemma
    \ref{alem:substc_snd}) and use our assumption on the typing of
    $c'$.  Then, the goal becomes $\jdgmt{\Delta}{\Psi'}{ \cbind{
        \hndl{ \hndlseq{c}{\Theta'} }{h}{e} }{A''} }$.  Thus, by the
    handling principle (Lemma \ref{alem:hndl_snd}) and our
    assumptions, the goal changes to
    $\jdgmt{\Delta}{\Psi_h}{\cbind{\hndlseq{c}{\Theta'}}{A'}}$ and
    follows by induction.
  \end{description}
\end{proof}

\begin{lemma}[Eval principle]\label{alem:eval_snd}
  If $\jdgmt{\Delta}{\ectxsym}{\cbind{c}{A}}$, then
  $\pjdgmt{\Delta}{\heval{c}}{A}$.
\end{lemma}
\begin{proof}
  By induction on $c$.
  \begin{description}
  \item[Case $c = \ret{e}$.]
    The goal is to show $\pjdgmt{\Delta}{e}{A}$ and it follows by
    inversion on the typing of $\ret{e}$ and weakening of $\ectxsym$ into
    $\Gamma$.
  \item[Case $c = \compappl{\textit{op}}{e}{x}{c'}$]
    contradicts typing as the context of effects is empty.
  \item[Case $c = \contappl{k}{e_1}{e_2}{x}{c'}$]
    also contradicts typing.
  \item[Case $c = \sappl{\shandle{u}{\Theta}{h}{e}}{x}{c}$.]
    The goal is to show $\jdgmt{\Delta}{\Gamma}{\tbind{
        \teval{\Theta, \hseqclause{h}{e}{x}{c}}{u}
      }{A}
    }$, while we assume $\mhypbind{u}{A'}{\Psi'} \in \Delta$,
    $\jdgmt{\Delta}{\ectxsym}{\hbind{h}{A'}{\Psi'}{S}{A''}}$,
    $\pjdgmt{\Delta}{e}{S}$, and
    $\jdgmt{\Delta, \tbind{x}{A''}}{\ectxsym}{\cbind{c}{A}}$.
    By the rule \textsc{seqh}, we obtain $\jdgmt{\Delta}{\ectxsym}{
      \hseqbind{
        \Theta, \hseqclause{h}{e}{x}{c}
      }{A'}{\Psi'}{A}
    }$ and this, together with the typing of $u$, is sufficient
    as a premise to the rule \textsc{eval} to show the goal.
  \item[Case $c = \letbox{u}{e}{c}$.]
    The goal is to show $\pjdgmt{\Delta}{\letbox{u}{e}{\heval{c}}}{A}$ and the
    goal follows with the rules $\Box E$, $\Box E$-comp, weakening, and
    induction.
  \item[Case $c = \tfix{f}{x}{A'}{\Psi}{c_1}{c_2}$.]
    Similarly to the case above, this case follows by the rules \textsc{fix},
    \textsc{fix-comp}, weakening, and induction.
  \end{description}
\end{proof}

\begin{lemma}[Modal substitution principle]\label{alem:substm_snd}
  If $\jdgmt{\Delta}{\Psi}{\cbind{c}{A}}$
  and $\jdgmt{\Delta, \mhypbind{u}{A}{\Psi}}{\Gamma}{t : J_{\mathit{term}}}$,
  then $\jdgmt{\Delta}{\Gamma}{\substm{t}{\Psi}{c}{u} : J_{\mathit{term}}}$.
\end{lemma}
\begin{proof}
  By mutual induction on the expansions of the second premise.
  The interesting cases are the substitutions for $u$ in \texttt{handle} and
  \texttt{eval}:
  \begin{description}
  \item[Case $\jdgmt{\Delta,
      \mhypbind{u}{A}{\Psi}}{\Gamma}{\cbind{\happl{u}{\Theta}{h}{e}{x'}{c'}}{B}}$.]
    According to the definition of the modal substitution, our goal is
    to show $\jdgmt{\Delta}{\Gamma}{ \cbind{ \substc{ \hndl{
            \hndlseq{c}{\substm{\Theta}{\Psi}{c}{u}}
          }{\substm{h}{\Psi}{c}{u}}{\substm{e}{\Psi}{c}{u}}
        }{x'}{(\substm{c'}{\Psi}{c}{u})} }{B} }$, and we have the
    assumptions $\jdgmt{\Delta,
      \mhypbind{u}{A}{\Psi}}{\Psi'}{\hseqbind{\Theta}{A}{\Psi}{A'}}$,
    $\jdgmt{\Delta,
      \mhypbind{u}{A}{\Psi}}{\Gamma}{\hbind{h}{A'}{\Psi'}{S}{A''}}$,
    $\pjdgmt{\Delta, \mhypbind{u}{A}{\Psi}}{e}{S}$, and
    $\jdgmt{\Delta, \mhypbind{u}{A}{\Psi},
      \tbind{x'}{A''}}{\Gamma}{\cbind{c'}{B}}$.  Then, we apply the
    monadic substitution principle (Lemma \ref{alem:substc_snd}) with
    induction on our assumption on the typing of $c'$.  The goal
    becomes $\jdgmt{\Delta}{\Gamma}{ \cbind{ \hndl{
          \hndlseq{c}{\substm{\Theta}{\Psi}{c}{u}}
        }{\substm{h}{\Psi}{c}{u}}{\substm{e}{\Psi}{c}{u}} }{A''} }$.
    By the handling principle (Lemma \ref{alem:hndl_snd}), and induction over the
    typing of $h$ and $e$, the goal changes to
    $\jdgmt{\Delta}{\Psi}{\cbind{\hndlseq{c}{\substm{\Theta}{\Psi}{c}{u}}}{A'}}$
    and follows by the theorem \ref{alem:hndl_seq_snd} and induction
    over $\Theta$.
  \item[Case $\pjdgmt{\Delta, \mhypbind{u}{A}{\Psi}}{\teval{\Theta}{u}}{B}$.]
    The goal is to show $\pjdgmt{\Delta}{
        \heval{\hndlseq{c}{\substm{\Theta}{\Psi}{c}{u}}}
    }{B}$, while we assume
    $\jdgmt{\Delta, \mhypbind{u}{A}{\Psi}}{\ectxsym{}}{
      \hseqbind{\Theta}{A}{\Psi}{B}
    }$.
    By induction on $\Theta$, we obtain $\jdgmt{\Delta}{\ectxsym{}}{
      \hseqbind{\substm{\Theta}{\Psi}{c}{u}}{A}{\Psi}{B}
    }$ and we use this as a premise to the handling sequencing principle lemma
    \ref{alem:hndl_seq_snd} to obtain $\jdgmt{\Delta}{\ectxsym{}}{
      \cbind{\hndlseq{c}{\substm{\Theta}{\Psi}{c}{u}}}{B}
    }$.
    Finally, we use this assumption as a premise to the eval principle lemma
    \ref{alem:eval_snd} to show the goal.
  \end{description}
\end{proof}

\begin{lemma}[Identity handler]\label{alem:id_handler_snd}
  $\jdgmt{\Delta}{\Psi}{\hbind{\handid{\Psi}}{A}{\Psi}{\unitt{}}{A}}$.
\end{lemma}
\begin{proof}
  By induction on $\Psi$.
  \begin{description}
  \item[Case $\Psi = \ectxsym{}$.]
    The goal is to show $\jdgmt{\Delta}{\ectxsym{}}{\hbind{
        \handbase{x}{z}{\ret{x}}
    }{A}{\ectxsym{}}{\unitt{}}{A}}$.
    The goal immediately follows by the \textsc{reth} and \textsc{ret} rules.
  \item[Case $\Psi = \Psi', \chypbind{op}{B}{B'}$.]
    The goal is to show $\jdgmt{\Delta}{\Psi', \chypbind{op}{B}{B'}}{
      \hbind{
        \handid{\Psi'}, \handoparr{\textit{op}}{x}{k}{z}{
          \sappl{
            \sop{\textit{op}}{x}
          }{y}{\sscont{k}{y}{z}}
        }
      }{A}{\Psi', \chypbind{op}{B}{B'}}{\unitt{}}{A}
    }$.
    By the rule \textsc{oph}, we have to show
    $\jdgmt{\Delta}{\Psi',
      \chypbind{\textit{op}}{B}{B'}}{\hbind{\handid{\Psi'}}{A}{\Psi'}{\unitt{}}{A}}$
    which follows by weakening on the inductive hypothesis.
    Then, we also must show $\jdgmt{\Delta, \tbind{x}{B},
      \tbind{z}{\unitt{}}}{\Psi', \chypbind{\textit{op}}{B}{B'},
      \khypbind{k}{B'}{\unitt{}}{A}}{
      \cbind{
        \sappl{
          \sop{\textit{op}}{x}
        }{y}{\sscont{k}{y}{z}}
      }{A}
    }$ which follows by rules \textsc{bind}, \textsc{op}, \textsc{cont}, and
    \textsc{ret}.
  \end{description}
\end{proof}

\subsection{Main Theorems}

\begin{theorem}[Local soundness]\label{alem:beta_red}
If $\jdgmt{\Delta}{\Psi}{\cbind{c}{A}}$, then the following
$\beta$-reductions are well typed.
\begin{enumerate}
\item If $\jdgmt{\Delta}{\Gamma}{\cbind{\letboxu{\tbox{\Psi}{c}}{c'}}{B}}$,
  then $\jdgmt{\Delta}{\Gamma}{\cbind{\substm{c'}{\Psi}{c}{u}}{B}}$.
\item If $\pjdgmt{\Delta}{\letboxu{\tbox{\Psi}{c}}{e'}}{B}$,
  then $\pjdgmt{\Delta}{\substm{e'}{\Psi}{c}{u}}{B}$.
\end{enumerate}
\end{theorem}
\begin{proof}
  Immediately follows from inversion on the \texttt{let-box} premises and Lemma
  \ref{alem:substm_snd}.
\end{proof}

\begin{theorem}[Local completeness]\label{alem:eta_box_snd}
  If $\pjdgmt{\Delta}{e}{\boxt{\Psi}{A}}$,
  then the $\eta$-expansion of $e$ is well-typed, i.e.
  $\pjdgmt{\Delta}{
    \letbox{u}{e}{
      \tbox{\Psi}{
        \shandlesimp{u}{\handid{\Psi}}{\tunit{}}
      }
    }
  }{\boxt{\Psi}{A}}$.
\end{theorem}
\begin{proof}
  By the \textsc{$\Box E$} rule and our assumption on the typing of $e$, we need
  to show
  $\pjdgmt{\Delta, \mhypbind{u}{A}{\Psi}}{
    \tbox{\Psi}{
      \sappl{
        \shandlesimp{u}{\handid{\Psi}}{\tunit{}}
      }{x}{\ret{x}}
    }
  }{\boxt{\Psi}{A}}$.
  Thus, by the \textsc{$\Box I$} rule, we need to show
  $\pjdgmt{\Delta, \mhypbind{u}{A}{\Psi}}{
    \sappl{
      \shandlesimp{u}{\handid{\Psi}}{\tunit{}}
    }{x}{\ret{x}}
  }{A}$.
  By the rules \textsc{ret} and \textsc{bind}, it suffices to show
  $\jdgmt{\Delta, \mhypbind{u}{A}{\Psi}}{\Psi}{
    \sbind{
      \shandlesimp{u}{\handid{\Psi}}{\tunit{}}
    }{A}
  }$.
  The goal follows by the rule \textsc{mvar} and the identity handler lemma
  \ref{alem:id_handler_snd}.
\end{proof}

\begin{theorem}[Preservation on expressions]\label{athm:preservation_exp}
  If $\vdash {\tbind{e}{A}}$ and $e \reduces{} e'$, then
  $\vdash {\tbind{e'}{A}}$.
\end{theorem}
\begin{proof}
By induction over $e \reduces{} e'$. The cases are:
\begin{description}
\item[Case $\appl{e_1}{e_2}~\reduces{}~\appl{e_1'}{e_2}$.]  We need to
  show $\vdash {\tbind{\appl{e_1'}{e_2}}{A}}$, while assuming $e_1
  \reduces{} e_1'$ and $\vdash {\tbind{\appl{e_1}{e_2}}{A}}$. 
  The goal follows by inversion on the second assumption,
  \textsc{${\rightarrow} E$} rule and induction.

\item[Case $\appl{v}{e_2}~\reduces{}~\appl{v}{e_2'}$.]
  Follows analogously.
\item[Case $\appl{(\lamx{A}{e_1})}{e_2} ~\reduces{}~
  \subst{e_1}{e_2}{x}$.]  Follows by inversion on the typing of $e$
  and the expression substitution principle (Lemma
  \ref{alem:subst_expr}).
\item[Case $\letboxu{e_1}{e_2} ~\reduces{}~ \letboxu{e_1'}{e_2}$.]
  The goal follows by inversion on the typing of $e$, the \textsc{$\Box E$}
  rule, and induction.
\item[Case $\letboxu{\tbox{\Psi}{c}}{e} ~\reduces{}~
  \substm{e}{\Psi}{c}{u}$.]  Follows by inversion on the typing of $e$ and
  the modal substitution principle (Lemma \ref{alem:substm_snd}).
\item[Case $e = 
    \tfix{f}{x}{B}{\Psi}{c}{e_1}
  $]
  which reduces to
  \[
  \subst{e_1}{
    \lam{x}{B}{
      \tbox{\Psi}{
        \tfix{f}{x}{B}{\Psi}{c}{c}
      }
    }
  }{f}
  \]
  The goal is to show
  $\pjdgmt{\ectxsym}{
    \subst{e_1}{
      \lam{x}{B}{
        \tbox{\Psi}{
          \tfix{f}{x}{B}{\Psi}{
            c
          }{c}
        }
      }
    }{f}
  }{A}$,
  under the assumptions that (1)
  $\pjdgmt{\tbind{f}{\funt{B}{\boxt{\Psi}{C}}}}{e_1}{A}$
  and (2) $\jdgmt{\tbind{f}{\funt{B}{\boxt{\Psi}{C}}},
    \tbind{x}{B}}{\Psi}{\cbind{c}{C}}$ for some type $C$. (These
  assumptions are obtained by inversion on the typing of $\vdash e :
  A$.)

  The proof follows immediately by applying the substitution principle
  (Lemma~\ref{lem:subst_expr}) on the variable $f$ in assumption
  (1). For the principle to apply, we must show that 
  \[
  \pjdgmt{\ectxsym{}}{
    \lam{x}{B}{
      \tbox{\Psi}{
        \tfix{f}{x}{B}{\Psi}{c}{c}
      }
    }
  }{\funt{B}{\boxt{\Psi}{C}}}
  \]
  but this is easy to establish from assumption (2) by the
  \textsc{${\rightarrow} I$}, \textsc{$\Box I$} and \textsf{Fix}
  typing rules.
\end{description}
\end{proof}

\begin{theorem}[Preservation on computations]
  If $\vdash {\cbind{c}{A}}$ and $c \reduces{} c'$, then
  $\vdash {\tbind{c'}{A}}$.
\end{theorem}
\begin{proof}
  By case analysis over $c \reduces{} c'$. The cases are:
  \begin{description}
  \item[Case $\ret{e}~\reduces{}~\ret{e'}$.]  We need to show $\vdash
    {\cbind{\ret{e'}}{A}}$, while assuming $e \reduces{} e'$.  By the
    preservation on expressions (Theorem \ref{athm:preservation_exp})
    and inversion on the typing of $\ret{e}$ we obtain $\vdash
    {\tbind{e'}{A}}$.  Thus the goal follows by the \textsc{ret} rule.
  \item[Case $\letbox{u}{e}{c''}~\reduces{}~\letbox{u}{e'}{c''}$.]  In
    this case, under assumption that $e \reduces{} e'$, we need to
    show that $\vdash {\cbind{\letbox{u}{e'}{c''}}{A}}$. By the
    preservation on expressions (Theorem \ref{athm:preservation_exp})
    and inversion on the typing of $\letbox{u}{e}{c''}$ we obtain
    $\vdash {\tbind{e'}{A'}}$ for some $A'$.  Thus the goal follows by
    the \textsc{$\Box E$} rule.
  \item[Case
    $\letbox{u}{\tbox{\Psi}{c_1}}{c_2}~\reduces{}~\substm{c_2}{\Psi}{c_1}{u}$.]
    We need to show $\vdash {\cbind{\substm{c_2}{\Psi}{c_1}{u}}{A}}$.  The
    goal follows by inversion on the typing of $c$ and the modal
    substitution principle (Lemma \ref{alem:substm_snd}).
  \item[Case $
      \tfix{f}{x}{B}{\Psi}{c_1}{c_2}
    $.]
    which reduces to
    \[
    \subst{c_2}{
      \lam{x}{B}{
        \tbox{\Psi}{
          \tfix{f}{x}{B}{\Psi}{c_1}{c_1}
        }
      }
    }{f}
    \]
    This case proceeds analogously to the one for expressions.
  \end{description}
\end{proof}

\begin{theorem}[Progress on expressions]\label{athm:progress_exp}
  If $\vdash {\tbind{e}{A}}$, then either
  (1) $e$ is a value, or
  (2) there exists $e'$ s.t. $e \reduces{} e'$.
\end{theorem}
\begin{proof}
By induction on $\vdash {\tbind{e}{A}}$.
Cases are:
\begin{description}
\item[Case $e = x$] contradicts typing.
\item[Case $e = \lam{x}{A}{e'}$] is a value.
\item[Case $e = \appl{e_1}{e_2}$.] By induction on $e_1$.
  If it is a value $v_1$, then by induction on $e_2$.
  If it is also a value $v_2$, then from typing and by inversion we know that
  $v_1$ must be of the form $\lam{x}{B}{e'}$ for some type $B$.
  Thus $e$ reduces to $\subst{e'}{v_2}{x}$ by the semantic step relation.
  If $e_2$ is not a value, but there exists $e_2'$ s.t. $e_2 \reduces{} e_2'$,
  then $e$ reduces to $\appl{v_1}{e_2'}$ by the semantic step relation.
  If $e_1$ is not a value, but there exists $e_1'$ s.t. $e_1 \reduces{} e_1'$,
  then $e$ reduces to $\appl{e_1'}{e_2}$ by the semantic step relation.
\item[Case $e = \tbox{\Psi}{c}$] is a value.
\item[Case $e = \letbox{u}{e_1}{e_2}$.] By induction on $e_1$, it is either a
  value or there exists $e_1'$ s.t. $e_1 \reduces{} e_1'$.
  In the former case, $e_1$ must be of the form $\tbox{\Psi}{c}$ from typing
  and by inversion.
  Thus in this case $e$ reduces to $\substm{e_2}{{\Psi}c}{u}$ by the semantic step
  relation.
  In the latter case, $e$ reduces to $\letbox{u}{e_1'}{e_2}$ by the semantic
  step relation.
\item[Case $e = \teval{\Theta}{u}$] contradicts typing since $u$ is a variable.
\item[Case $c = \tfix{f}{x}{B}{\Psi}{c}{e}$.] $c$ reduces to
  \[
  \subst{e}{\lam{x}{B}{\tbox{\Psi}{
        \tfix{f}{x}{B}{\Psi}{c}{c}
  }}}{f}
  \]
  by the reduction relation.
\end{description}
\end{proof}

\begin{theorem}[Progress on computations]
  If $\vdash {\cbind{c}{A}}$, then either
  (1) $c$ is $\ret{v}$ where $v$ is a value, or
  (2) there exists $c'$ s.t. $c \reduces{} c'$.
\end{theorem}
\begin{proof}
  By induction on $\vdash {\cbind{c}{A}}$.
  Cases are:
  \begin{description}
  \item[Case $c = \ret{e}$.] By inversion on $c$, and the progress on expressions
    theorem \ref{athm:progress_exp}, two cases possible on $e$: either it is a
    value, by which $c$ is $\ret{v}$ and the first statement holds, or there
    exists $e'$ s.t. $e \reduces{} e'$.
    Then, $c$ reduces to $\ret{e'}$ and the second statement holds.
  \item[Case $c = \sappl{s}{x}{c'}$] contradicts typing, as all statements use
    variables.
  \item[Case $c = \tfix{f}{x}{B}{\Psi}{c_1}{c_2}$.] $c$ reduces to
    \[
    \subst{c_2}{\lam{x}{B}{\tbox{\Psi}{
          \tfix{f}{x}{B}{\Psi}{c_1}{c_1}
    }}}{f}
    \]
    by the reduction relation.
  \item[Case $c = \letbox{u}{e}{c_1}$.] By inversion on $c$, and the theorem
    \ref{athm:progress_exp}, $e$ is either a value or there exists $e'$ s.t. $e
    \reduces{} e'$.
    In the first case, $e$ must be of a form $\tbox{\Psi}{c_b}$ for some
    $\Psi$, by typing and inversion.
    Thus $c$ reduces to $\substm{c_1}{\Psi}{c_b}{u}$.
    In the second case, $c$ reduces to $\letbox{u}{e'}{c_1}$.
  \end{description}
\end{proof}

\end{document}